\documentclass[%
 reprint,
 twocolumn,
superscriptaddress,
longbibliography,
preprintnumbers,
 amsmath,amssymb,
 aps,
pra,
]{revtex4-2}

\usepackage{blochsphere}
\usepackage{graphicx}
\usepackage{tabularx}
\graphicspath{{images/}}
\usepackage[]{xcolor} 
\usepackage{xcolor, colortbl}
\usepackage{verbatimbox}
\usepackage{array}
\usepackage{epsfig,color}
\usepackage[most]{tcolorbox}
\usepackage{soul} 

\usepackage{algorithm}
\usepackage{algpseudocode}

\usepackage{physics}
\usepackage{amsmath}
\usepackage{amssymb}
\usepackage{amsthm} 
\usepackage{amsfonts}
\usepackage{mathtools} 
\usepackage{times}
\usepackage{listings} 
\usepackage{thmtools} 
\usepackage{thm-restate} 
\usepackage{dcolumn}
\usepackage{bm}
\usepackage[colorlinks=true, linkcolor={urlblue},citecolor={urlblue},urlcolor={urlblue}]{hyperref}
\usepackage[sans]{dsfont} 
\usepackage{capt-of}
\usepackage[normalem]{ulem} 
\usepackage[caption=false]{subfig}
\newsavebox{\imagebox}
\usepackage[capitalise]{cleveref}
\usepackage{enumitem}
\usepackage[T1]{fontenc} 
\UseRawInputEncoding

\definecolor{nblue}{rgb}{0.2,0.2,0.7}
\definecolor{ngreen}{rgb}{0.2,0.6,0.2}
\definecolor{nred}{rgb}{0.7,0.2,0.2}
\definecolor{nblack}{rgb}{0,0,0}
\definecolor{urlblue}{RGB}{30,19,156}

\definecolor{tssteelblue}{RGB}{70,130,180}
\definecolor{tssteelorange}{RGB}{161,57,65}
\definecolor{tsorange}{RGB}{255,138,88}
\definecolor{tsblue}{RGB}{23,74,117}
\definecolor{tsforestgreen}{RGB}{21,122,81}
\definecolor{tsyellow}{RGB}{255,185,88}
\definecolor{tsgrey}{RGB}{200,200,200}

\definecolor{mangotango}{rgb}{1.0, 0.51, 0.26}

\lstdefinestyle{mystyle}{
	breaklines=true,                 
}
\def\G{\mathcal{G}}

\def\bOmega{\bm{\Omega}}

\def\C{\mathbf{C}}

\def\I{\mathbf{I}}

\def\N{\mathbf{N}}
\def\M{\mathbf{M}}

\def\T{\mathbf{T}}

\def\X{\mathbf{X}}

\def\V{\mathbf{V}}

\newcommand{\zero}{\mathbf{0}} 

\def\d{\mathbf{d}}

\def\ax{\overrightarrow{\mathbf{x}}}

\def\bx{\overline{\mathbf{x}}}
\def\hx{\widehat{x}}
\def\hq{\widehat{q}}
\def\hp{\widehat{p}}

\def\vari{\text{Var}}

\renewcommand{\t}{^{\mbox{\tiny T}}}

\newcommand{\+}{^{\dagger}} 

\def\with{\text{with}}

\def\gat{Gaussian-atemporality}
\def\gar{Gaussian atemporality robustness}

\theoremstyle{definition}
\newtheorem{dfn}{Definition}
\newtheorem{thm}{Theorem}
\newtheorem{lem}[thm]{Lemma}

\newtheorem{obs}{Observation}

\newtheorem*{con*}{Conjecture}

\theoremstyle{remark}

\begin{document}

\title{Gaussian Atemporality: When Gaussian Quantum Correlations Imply Common Cause}

\author{Minjeong Song}
\email{song.at.qit@gmail.com}
\affiliation{Nanyang Quantum Hub, School of Physical and Mathematical Sciences, Nanyang Technological University, 637371, Singapore}
\affiliation{Centre for Quantum Technologies, National University of Singapore, 117543, Singapore}
\author{Jayne Thompson}%
\affiliation{College of Computing and Data Science, Nanyang Technological University}
\affiliation{Centre for Quantum Technologies, Nanyang Technological University, 637371, Singapore}
\author{Matthew S. Winnel}
\affiliation{Centre for Quantum Computation and Communication Technology, School of Mathematics and Physics, University of Queensland, St Lucia, Queensland 4072, Australia}
\author{Biveen Shajilal}
\affiliation{Quantum Innovation Centre (Q.InC), Institute of Materials Research and Engineering (IMRE), Agency for Science, Technology and Research (A$^\star$STAR), 138634, Singapore.}
\author{Timothy C. Ralph}
\affiliation{Centre for Quantum Computation and Communication Technology, School of Mathematics and Physics, University of Queensland, St Lucia, Queensland 4072, Australia}%
\author{Syed M. Assad}
\affiliation{Quantum Innovation Centre (Q.InC), Institute of Materials Research and Engineering (IMRE), Agency for Science, Technology and Research (A$^\star$STAR), 138634, Singapore.}
\author{Mile Gu}
\email{mgu@quantumcomplexity.org}
\affiliation{Nanyang Quantum Hub, School of Physical and Mathematical Sciences, Nanyang Technological University, 637371, Singapore}
\affiliation{Centre for Quantum Technologies, Nanyang Technological University, 637371, Singapore}

\begin{abstract} 

Conventionally, covariances do not distinguish between spatial and temporal correlations. The same covariance matrix could equally describe temporal correlations between observations of the same system at two different times or correlations made on two spatially separated systems that arose from some common cause. Here, we demonstrate Gaussian quantum correlations that are \emph{atemporal}, such that the covariances governing their quadrature measurements are unphysical without postulating some common cause. We introduce \gar\ as a measure of atemporality, illustrating its efficient computability and operational meaning as the maximum noise which can be added without removing this uniquely quantum phenomenon. We illustrate that (i) specific spatiotemporal Gaussian correlations possess an intrinsic arrow of time, such that \gar\ is zero in one temporal direction and not the other and (ii) that it measures quantum correlations beyond entanglement. 
\end{abstract}

\maketitle

\section{Introduction}
Multivariate Gaussian distributions are used almost universally in statistical sciences, providing a first-order description of correlations. Given random variables $X_A$ and $X_B$ that are post-processed to have zero means, the distribution is entirely specified by the covariance matrix $V_{\alpha,\beta} = \langle X_\alpha X_\beta \rangle$, $\alpha, \beta \in \{A,B\}$ -- and can be systematically estimated through repeated sampling. However, as per the iconic idiom `correlation does not imply causation', the correlation matrix alone cannot tell us anything about underlying causal relations between $X_A$ and $X_B$. 

\begin{figure*}[ht]
  \includegraphics[width=0.7\linewidth]{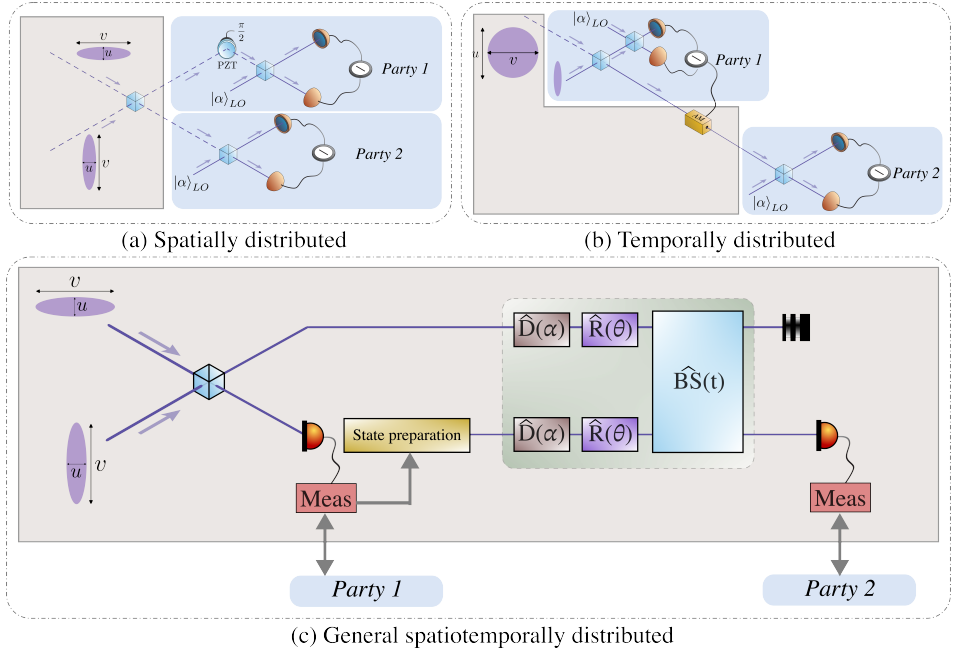}
  \caption{\textbf{Spatio-temporal causal distribution mechanisms.} The blue-shaded regions in the figure show the measurement stations of the parties involved, Alice and Bob. The respective parties perform Gaussian measurements, i.e., homodyne or heterodyne measurements. For simplicity of representation, a) and b) show homodyne measurements in each station, and parties are subject to any allocation for these station; Party 1 could be either Alice or Bob. Likewise, Party 2 could be Bob or Alice. This allows the above diagrammatic descriptions not to be confined by a specific causal order of Alice's and Bob's measurement stations. Consequently, the causal mechanisms can be classified into three types: Fig.~a) shows the spatially distributed mechanism where parties 1 and 2 perform the measurement simultaneously on a two-mode squeezed vacuum (TMSV) state. Party 1 can change the measurement basis by imparting an additional phase through the piezoelectric transducer~(PZT). Fig.~b) shows the temporally distributed mechanism where Party 1 performs continuous variable quantum non-demolition measurements~\cite{andersen2002nondemolotion,buchler2001nondemolition} on the input state and shares the residual state with Party 2. 
  The AM is the amplitude modulator, which is a part of the quantum non-demolition measurement. 
  Fig.~c) shows the general spatiaotemporally distributed mechanism that is practically implementable~\cite{Note_process}. The "Meas" represent the Gaussian measurements. The state preparation box represents Gaussian operations that are generally used to prepare states, namely the displacement operation~$\mathrm{\hat{D}}(\alpha)$, rotation~$\mathrm{\hat{R}}(\theta)$, and squeezing operation~$\mathrm{\hat{S}}(r)$, and it depends on the measurement basis and the outcome. 
  The displacement operation and rotation operation, along with the beamsplitting operation $\mathrm{\hat{BS}}(t)$, constitute the elements of the arbitrary two-mode unitary operation between the ancilla and one arm of the TMSV state.}
  \label{fig:cv_structure}
\end{figure*}

Consider the thought experiment where Alice and Bob are each located in their laboratories, such that one receives a system $A$ and the other a system $B$ that they choose to observe in two different ways to obtain respective random variables $X$ and $Y$ such that the resulting measurement statistics is a multivariate Gaussian. They present us with all resulting covariances $\langle X_\alpha X_\beta \rangle$, $\langle Y_\alpha Y_\beta \rangle$ and $\langle X_\alpha Y_\beta \rangle$, and three possible \emph{causal mechanisms} to explain observed correlations:
\begin{itemize}
    \item[(a)] Purely spatial -  where Alice and Bob receive two arms of an initially correlated system, corresponding to the interpretation of common cause (see Fig.~1a).
    \item[(b)] Purely temporal - where Alice and Bob are measuring the same system at two different times, corresponding to the interpretation of direct cause (see Fig.~1b).
    \item[(c)] A combination of the two mechanisms above, involving both direct cause and common cause contributions (see Fig.~1c).
\end{itemize}
Could we rule out any particular mechanism? Classical statistics say this is impossible~\cite{ried15advantage}. Any covariance matrix observed is consistent with all possible causal mechanisms.

However, quantum systems can exhibit significant differences. When Alice and Bob each receive a Gaussian state encoded within a quantum harmonic oscillator (e.g., a quantum mode of light, henceforth referred to as a qumode), all quadrature measurement statistics between Alice and Bob are described by a multivariable Gaussian. Yet, certain measurement correlations are \emph{aspatial} - ruling out explanation through purely spatial mechanisms~\cite{zhang20different}.

Here we ask the complementary question: what conditions imply covariances are \emph{atemporal}, such that they require us to postulate that $A$ and $B$ share a common cause? In doing this, we introduce \gat\ as incompatibility with a temporal causal structure. We then show that (i) for standard spatial two-mode continuous-variable states, \gat\ provides a measure of non-classicality distinct from entanglement, and (ii) \gat\ exhibits asymmetry under time reversal, such that certain two-mode spatial-temporal Gaussian correlations have an intrinsic causal arrow of time.

When combined with previous results on aspatial Gaussian spatiotemporal states, our results allow full classification of Gaussian quantum correlations based on their compatibility with various causal mechanisms. This mirrors recent developments in the study of spatiotemporal correlations on qubits \cite{ried15advantage,fitzsimons15pdo,song2024class, song2025bipartite, zhao18pdogeometry,hu18discrimination,marletto19pdoentanglement,pisarczyk19pdocausallimit,marletto2020non,zhang20observationscheme,marletto21temporaltele,utagi2021pdomarkov,jia23marginal,liu2024arrowoftime,liu2025lighttouch,fullwood2025pdo,liu2025randomized}
. Meanwhile we illustrate several unique advantages of studying spatiotemporal correlations in Gaussian regime. First, it facilitates \emph{\gar}, an operational and computable measure of \gat\ that captures the maximum amount of Gaussian noise we can add without breaking \gat. Second, it enables us to determine atemporality analytically without numerical minimization, and thus enables an explicit expression of \gar.

\section{Preliminaries}
\subsection{Notation and Framework}

We first introduce core concepts and notation. We use hat notation for infinite-dimensional operators. Given a quantum mode (qumode), we denote associated `position' operator and `momentum' operators respectively by $\hq$ and $\hp$. The canonical commutation relations read~\footnote{It is the widely accepted convention to omit the identity operator from the canonical commutation relation.}  
$$[\hq,\hp]= i\hbar := 2i,$$ 
whereby we adopt the natural units $\hbar:=2$. We consider two single-mode Gaussian systems, each labeled by $A$ and $B$ for brevity, however, our results can be generalized to any bipartite Gaussian systems consisting of $N$ modes for any finite $N$. The quadrature operators for the composite systems $AB$ are represented by an operator-valued vector 
$$\ax_{AB} \equiv (\hq_A,\hp_A,\hq_B,\hp_B)\t,$$ 
whose 
$j^{\mathrm{th}}$ component is denoted by $\hx_{AB,j}$, and $(\cdot)\t$ represents the transpose map. Whenever the context is clear, we omit the subscript $_{AB}$. The canonical commutation relation among these quadrature operators can be concisely written as 
$$ [\ax,\ax\t]= 2i\bOmega_2, \quad \text{with } \bOmega_2 \equiv (\begin{smallmatrix}
\bOmega & \zero \\
\zero & \bOmega
\end{smallmatrix}),$$ 
where $[\ax,\ax\t]$ represents the commutator of respective entities of $\hx$ in a outer product form, i.e., $[\hx_j,\hx_k]= 2i\bOmega_{2,jk}$. The matrix $\bOmega$ denotes the symplectic form, the traceless anti-symmetric matrix, i.e., $\bOmega\equiv (\begin{smallmatrix}
0 & 1 \\
-1 & 0
\end{smallmatrix})$, and $\zero$ denotes the zero matrix. We also sometimes use $\zero$ to denote the zero vector, depending on the context. Boldface letters are used to specify finite-dimensional matrices and vectors. 

When the measurement statistics of quantum systems that are spatially distributed, any general state can be represented by a density matrix. In the case of qumodes, these become infinite-dimensional. Gaussian states characterize the subset that can be completely represented by a mean vector and a covariance matrix. The mean vector $\bx$ is the first statistical moment of quadrature operators, defined as $\bx_{j}\equiv \expval{\hx_j}$. The covariance matrix $\V$ is the second moment, defined as $V_{jk} \equiv \frac{1}{2}\expval{\Delta\hx_j \Delta\hx_k + \Delta\hx_k\Delta\hx_j}$, with $\Delta\hx_j \equiv \hx_j - \expval{\hx_j}$. Thus, if Alice and Bob communicate this information to us, we have tomographical knowledge of any spatial bipartite Gaussian state.

The covariance matrix (CM), with its sub-matrices provides a concise method to express this information. Formally, we define a CM by 
\begin{eqnarray}
    \V_{AB} &=& \begin{pmatrix}
    \V_A & \C \\
    \C\t & \V_B
    \end{pmatrix}, 
\end{eqnarray}
where $V_{A,jk} \equiv \frac{1}{2}\expval{\Delta\hx_{A,j} \Delta\hx_{A,k} + \Delta\hx_{A,k}\Delta\hx_{A,j}}$ and $V_{B,jk} \equiv \frac{1}{2}\expval{\Delta\hx_{B,j} \Delta\hx_{B,k} + \Delta\hx_{B,k}\Delta\hx_{B,j}}$ represent local CMs. We call $C_{jk}\equiv \frac{1}{2}\expval{\Delta\hx_{A,j} \Delta\hx_{B,k} + \Delta\hx_{B,k}\Delta\hx_{A,j}}$ the \emph{cross-correlation matrix} since it accounts for the correlations between two modes. Note that $\hx_{A,j}$ and $\hx_{B,j}$ are $j$th quadrature operators of local modes, $A$ and $B$, respectively, i.e., $\hx_{A,j} \in \{\hq_A, \hp_A\}$ and $\hx_{B,j} \in \{\hq_B, \hp_B\}$. Meanwhile $\V_A$ and $\V_B$ describe local measurement statistics at $A$ and $B$.

When considering bipartite correlations, Alice and Bob are free to choose the orientations in phase space that define position and momentum as well as the location of the origin. There then exists a choice of quadrature basis $\ax^\star_{AB}$ such that their mean vectors become trivial, i.e., $\bx^\star_A=\zero$ and $\bx^\star_B=\zero$, and 
their local CMs do not have off-diagonal terms, i.e., $\V_A^\star =(\begin{smallmatrix}
            \vari(\hq_A^\star) & 0 \\
            0 & \vari(\hp_A^\star)
        \end{smallmatrix})$ and $\V_B^\star = (\begin{smallmatrix}
            \vari(\hq_B^\star) & 0 \\
            0 & \vari(\hp_B^\star)
        \end{smallmatrix})$.
We refer to this as a standard basis. In practice, Alice (Bob) can always determine the standard basis by performing various quadrature measurements, while Bob (Alice) does nothing. In the following, we will suppose that Alice's and Bob's measurement bases are standard, and omit the superscript $^\star$ for brevity.

Once the standard basis is determined, the elements of the CM can be determined through a standard measurement procedure: Alice randomly chooses i) to do nothing, ii) to implement a non-demolition homodyne measurement in a basis $\hq_A$, or iii) in a basis $\hp_A$. Similarly, Bob also independently chooses i) to do nothing, ii) to implement a non-demolition homodyne measurement in a basis $\hq_B$, or iii) in a basis $\hp_B$. By repeating this procedure a sufficient number of rounds, they will be able to determine the elements of the CM to any desired level of accuracy.

\subsection{Spatial-temporal Gaussian states}

While the measurement procedure above typically assumes that Alice and Bob share a spatially distributed state, the same procedure can be applied  regardless of underlying causal mechanism. We now return to the general scenario where Alice and Bob are situated in separate laboratories, where they each receive a Gaussian continuous-variable state. Let their local quadrature operators be in the standard basis. Alice and Bob then perform the same measurement procedure as above, whereby the resulting covariance matrix defines a \emph{space-time statistical covariance matrix} in standard form. More formally:

\begin{dfn}[Space-time covariance matrix in standard form]
    Let Alice and Bob each have access to a qumode with standard basis $(\hq_A,\hp_A)$ and $(\hq_B,\hp_B)$, with locally diagonal covariance matrices $\V_A$ and $\V_B$. Let Alice and Bob then measure their quadrature operators at random, with resulting outcomes governed respectively by random variables $\hx_{A,j}$ and $\hx_{B,k}$ and the cross-correlation matrix $C_{jk} = \frac{1}{2}\expval{\hx_{A,j}\hx_{B,k}+\hx_{B,k}\hx_{A,j}}$. The space-time covariance matrix is then given by the matrix
    \begin{eqnarray}
    \V_{AB} &=& \begin{pmatrix}
    \V_A & \C \\
    \C\t & \V_B
    \end{pmatrix}, \label{eq:submat}
    \end{eqnarray}   
\end{dfn}

While the form looks identical to that of the standard covariance matrix, we note that this definition no longer assumes there is an underlying bipartite quantum state $\rho_{AB}$ shared between Alice and Bob. In particular, it is agnostic to underlying causal mechanisms. Therefore, they extend the conventional covariance matrix, and were originally proposed as a unified representation to describe space-time Gaussian states~\cite{zhang20different}. 

The local CMs, $\V_A,\V_B$ correspond to conventional CMs that describe local statistics. They are positive definite matrices that obey the uncertainty relation. However, the global space-time covariance matrix $\V_{AB}$ does not need to obey such constraints -- owing to the possibility of arising from temporal causal mechanisms. Consider a temporal causal structure where Alice receives an initial Gaussian state $\rho_A$ which is then sent to Bob through a Gaussian channel $\G:A \to B$. Alice measures the state before passing it through the channel $\G$, and Bob measures the subsequent output.
The local CMs, $\V_A$ and $\V_B$ are uniquely determined by the state received by Alice $\rho_A$ and the state received by Bob $\G(\rho_A)$. The cross-correlation matrix $\C$ then has components 
\begin{eqnarray*}
    C_{jk} = \int \dd{x_A}x_A \tr(\widehat{\Pi}^{(j)}_{x_A}\rho_A)\tr(\G(\widehat{\Pi}^{(j)}_{x_A})\hx_{B,k}),
\end{eqnarray*}
where $\widehat{\Pi}^{(j)}_{x_A}$ represents the projector associated with the eigenvalue $x_A$ of the quadrature operator $\hx_{A,j}$ \cite{zhang20different}. This corresponds to the following scenario: the initial state $\rho_A$ is measured with $\hx_{A,j}$, the state collapses to $\widehat{\Pi}^{(j)}_{x_A}$ based on the measurement outcome $x_A$, and it undergoes the channel $\G$. After that, $\G(\widehat{\Pi}^{(j)}_{x_A})$ is measured with $\hx_{B,k}$. Ref.~\cite{zhang20different} studied the case of the vacuum state evolving under the noiseless channel -- its corresponding covariance matrix is not positive definite and therefore is not physical spatial state. From this, we see that Gaussian correlations exist that are \emph{aspatial} - they are not compatible with any purely spatial causal mechanism. Here, we wish to determine conditions in which correlations are atemporal, such that no $\G$ exists. 

\section{Results}
When studying the scenario where Alice and Bob each received a qubit, specific Pauli measurement correlations were also atemporal: they cannot arise from any purely temporal causal mechanism. This then allowed a classification of all spatial-temporal correlations based on causal compatibility \cite{song2024class}. Here, we extend this to quadrature correlations on Gaussian states by asking: can continuous variable Gaussian systems exhibit \emph{Gaussian-atemporality}, such that their correlations cannot be explained by any Gaussian temporal causal mechanism?

\subsection{Temporally-Compatible Gaussian states}
We first examine the class of correlations compatible with temporal causal mechanisms. Consider the case where Alice measures her mode, which then evolves by some Gaussian channel $\G: A \to B$ before being given to Bob. In the standard basis, we can drop phase-space displacements. Such Gaussian channels can then be effectively defined by a transformation matrix $\T$ and a noise matrix $\N\ge\zero$ \cite{weedbrook12gaussian, holevo2001evaluating}, such that when a Gaussian state with covariance matrix $\V_{A}$ is injected as input, it emits a Gaussian state with covariance matrix
\begin{eqnarray}
    \V_{B} &= \T\V_{A}\T\t + \N. \label{eq:tn}
\end{eqnarray}
Furthermore, we prove the following lemma that allows us to write the resulting cross-correlation matrix between Alice and Bob in terms of $\T$ (see \cref{app:proofs} for the proof):

\begin{restatable}{lem}{temporal}
    Consider a temporal causal structure where an initial Gaussian state, characterized by a CM $\V_A$ evolves in time through a Gaussian channel $\G:A\to B$. Let $\T$ be a transformation matrix associated with $\G$. Then, the cross-correlation matrix $\C$ of the space-time CM,  i.e., $\V_{AB}\equiv ( \begin{smallmatrix}
\V_A & \C \\
{\C}\t & \V_B
\end{smallmatrix} )$, is given by 
     \begin{eqnarray*}
        \C = \V_A\T\t. 
     \end{eqnarray*} 
    \label{lem:temporal}
\end{restatable} 
Combining with Eq.~(\ref{eq:tn}), we obtain a complete \emph{closed-form expression} of the space-time CM $\V_{AB}$ of temporally distributed Gaussian states; given $\V_A, \T$ and $\N$, the sub-matrices $\V_{B}$ and $\C$ can be written as
     \begin{eqnarray}
     \begin{aligned}
        \V_{B} &= \T\V_{A}\T\t + \N, \qquad 
        \C &= \V_A\T\t. 
     \end{aligned}
     \end{eqnarray} 

Note that $\G$ is a valid Gaussian channel, represented by a completely positive (CP), Gaussian-preserving, and trace-preserving map. Therefore, its transformation and noise matrices, $\T$ and $\N$ must obey the CP condition \cite{serafini23qcv}
     \begin{eqnarray}
         \N+i\bOmega -i \T\bOmega\T\t \ge \zero. \label{eq:con_cp}
     \end{eqnarray} 
This then places a constraint on what space-time correlations are temporally compatible.

\subsection{Necessary and Sufficient Conditions for Atemporality}
The constraint above also provides a pathway to identify continuous-variable indicators of atemporality. Regardless of the underlying causal mechanism that generated a space-time covariance matrix $\V_{AB}$, we first assert that it is generated causally via a Gaussian channel from $A$ to $B$. This allows us to define a class of \emph{Gaussian pseudo-channels}:

\begin{dfn}[Forward Gaussian Pseudo-Channels]
Given space-time CM $\V_{AB}$ of two modes $A,B$ in standard form, an associated \emph{forward Gaussian pseudo-channel} $\overrightarrow{\Lambda} :A\to B$ is defined as a linear, Gaussian-preserving, and trace-preserving map satisfying the following conditions: 
\begin{eqnarray}
     \begin{aligned}
        \V_{B} &= \T\V_{A}\T\t + \N,\\
        \C &= \V_A\T\t, \label{eq:pc}
     \end{aligned}
     \end{eqnarray}
where $\T$ and $\N\ge\zero$ represent the transformation matrix and noise matrix, respectively, associated with $\overrightarrow{\Lambda}$. 
\end{dfn}

Inverting  Eq.~(\ref{eq:pc}) provides gives a systematic method to determine $\T$ and $\N$, and simultaneously illustrates that they are uniquely determined by $\V_{AB}$;

\begin{thm} Given a space-time CM $\V_{AB} \equiv ( \begin{smallmatrix}
\V_A & \C \\
{\C}\t & \V_B
\end{smallmatrix} )$, we can always find its \emph{unique} forward Gaussian pseudo-channel in the form of $\T,\N$ as follows,
    \begin{eqnarray}
    \begin{aligned}
        \T &= {\C}\t{\V^{-1}_A},\\
        \N &= \V_B - {\C}\t{\V^{-1}_A}\C.
     \end{aligned} 
     \end{eqnarray}   
Similarly, we can always find the unique reverse Gaussian pseudo-channel. \label{thm:cv_retrieve}
\end{thm}
\begin{proof}
    Recall $\vari(\hq_A)\vari(\hp_A)\ge 1$ due to the uncertainty principle, ensuring that $\V_A$ is invertible. Thus, the proof follows from inverting Eq.~(\ref{eq:pc}).
\end{proof}

If the resulting pseudo-channel is also completely positive, such that $\T$ and $\N$ satisfy the CP condition Eq.~(\ref{eq:con_cp}), then $\V_{AB}$ is clearly compatible with a temporal causal mechanism. On the other hand, if no such $\T$ and $\N$ exist, then no physical Gaussian channel from $A$ to $B$ is compatible with $\V_{AB}$. We then say that $\V_{AB}$ is \emph{forward Gaussian-atemporal}. This motivates us to define 
\begin{align}
    \X(\T,\N) \equiv \N+i\bOmega -i \T\bOmega\T\t    
\end{align}
as the \emph{forward atemporality matrix}. The non-positivity of $\X(\T,\N)$ then provides a necessary and sufficient condition for forward Gaussian-atemporality. By reversing $A$ and $B$, we can define the \emph{reverse Gaussian pseudo-channel} that attempts to explain $\V_{AB}$ as a channel from Bob to Alice. The state is \emph{reverse Gaussian-atemporal} should no CP pseudo-channels exist from $B$ to $A$. We then say that $\V_{AB}$ is Gaussian atemporal if it is Gaussian-atemporal in both temporal direction.

\subsection{Atemporality Robustness}

We introduce an operational measure of atemporality by taking inspiration from various robustness measures in quantum resources theories~\cite{chitambar2019quantum,napoli2016robustness}. \emph{Given a spatial-temporal Gaussian state that exhibits atemporality, what is the maximum noise we can add until it becomes compatible with a temporal causal mechanism}? 

Consider first the case of forward atemporality, where atemporality is equivalent to the condition $\X(\T,\N) = \N+i\bOmega -i \T\bOmega\T\t \ngeq \zero$, where $\N$ is the noise matrix. Clearly, if this noisy matrix is sufficiently large, $\X(\T,\N) \geq \zero$. i.e., for any $\X(\T,\N)$, there exists a sufficiently large $\mu$ such that $\X(\T,\N + \mu \I) \geq \zero$. Operationally, this corresponds to adding thermal noise of variance $\mu$ to Bob's measurement statistics. Generalizing thermal noises to any Gaussian noises, this motivates our definition of \emph{forward atemporality robustness}.

\begin{dfn} [Forward Atemporality Robustness] Consider a Gaussian spatial-temporal state on two qumodes with a space-time covariance matrix $\V_{AB}$. Let $\V'_{AB}$ represent the system's resulting covariance matrix should Bob's measurement outcome suffer an additive Gaussian noise $\mathbf{E}\ge\zero$. We define the \emph{forward atemporality robustness}, denoted $\overrightarrow{f}(\V_{AB})$, as the maximal amount of noise such that $\V'_{AB}$ is still atemporal. 
Mathematically,

\begin{align}
    \overrightarrow{f}(\V_{AB}) = \sup \{ \mathcal{N}(\mathbf{E}): \X(\T,\N + \mathbf{E}) < 0\}
\end{align}

where $X(\T,\N)$ is the forward atemporality matrix of $\V_{AB}$ and $\mathcal{N}$ is a measure of noise defined as 
\begin{align}
    \mathcal{N}(\mathbf{E}):=\sqrt{\det \mathbf{E}}.
\end{align}
\end{dfn}

First, observe that this measure is faithful. $\overrightarrow{f}(\V_{AB})= 0$ implies that $\X(\T,\N) \geq 0$, thus, $\V_{AB}$ is compatible with temporal causal explanations. Meanwhile, if $\overrightarrow{f}(\V_{AB}) > 0$, then $\X(\T,\N) < 0$, implying atemporality. Moreover, we illustrate that this measure can be computed analytically (see \cref{app:proofs}), enabling us to prove the following theorem:

\begin{restatable}{thm}{arobust} A spatial-temporal qumode-pair with covariance matrix $\V_{AB}$ and an associated forward pseudo-channel with transformation matrix $\T$ and noise matrix $\N$ has forward atemporality robustness
    \begin{eqnarray}
        \overrightarrow{f}(\V_{AB})=\max(0, \abs{\omega}-\sqrt{\det\N}),
    \end{eqnarray}
    with $\omega = 1-\det\T=1-\frac{\det\C}{\det \V_A}$ and $\det\N=\frac{\det\V_{AB}}{\det \V_A}$. Moreover, $\overrightarrow{f}(\V_{AB}) > 0$ is a necessary and sufficient condition for forward Gaussian atemporality.     
    \label{thm:arobust}
\end{restatable}

By exchanging $A$ and $B$, we can also define a \emph{reverse \gar} $\overleftarrow{f}$. We can then define general \emph{\gar} by taking the minimum of these two, i.e.,
\begin{eqnarray}
    f(\V_{AB}) \equiv \min\{\overrightarrow{f}(\V_{AB}),\overleftarrow{f}(\V_{AB}) \}.
\end{eqnarray}    
$f(\V_{AB}) > 0$ then provides a necessary and sufficient condition for Gaussian atemporarily.

\subsection{Entanglement and Atemporality} 
Consider first the class of spatially-compatible spatial-temporal two-mode Gaussian states, i.e., those that a bipartite density operator can represent. In this context, entanglement has always been the iconic measure of non-classicality. Meanwhile, the idiom that correlations do not imply causation from classical statistics suggests that highly classical states will likely be compatible with spatial and temporal causal mechanisms and thus have zero atemporality robustness.

\begin{figure}[tb]
    \centering
    \includegraphics[scale=0.5]{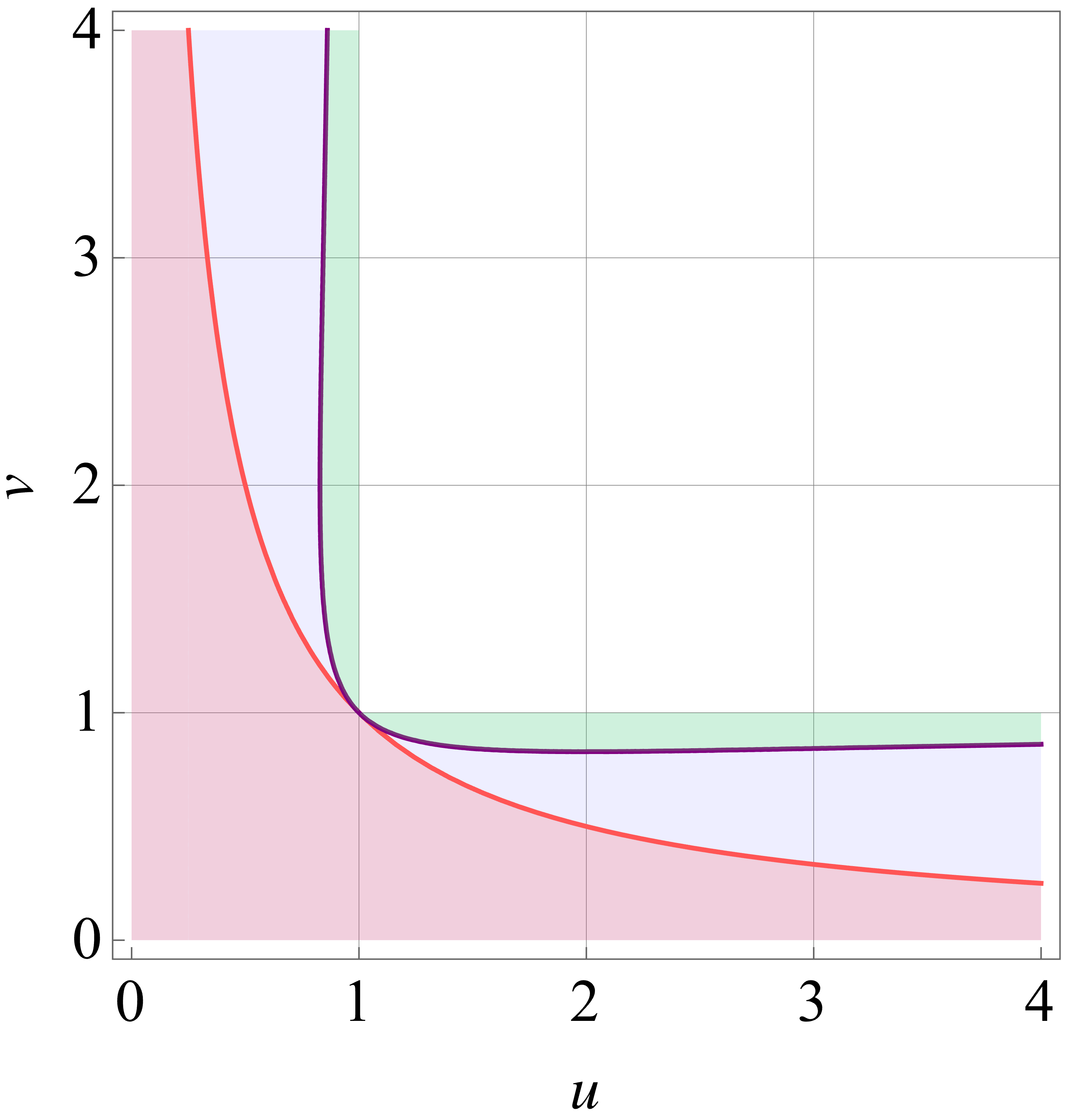}
    \caption{ \textbf{Entanglement and atemporality for asymmetric states passing through a balanced beam splitter.} The red-shaded region represents the set of all unphysical states. The regions where $u<1$ or $v<1$ represent the entangled states. States in the blue region are entangled and Gaussian-atemporal, whereas the states in the green region are entangled but not Gaussian-atemporal. The states in the remaining white region are separable and not Gaussian-atemporal.
    }
    \label{fig:ex_sym}
\end{figure}

We investigate these conjectures through examples. Consider two thermal states with variance $v\ge1$ labelled by $A$ and $B$, passing through a two-mode squeezing operation with a squeezing parameter $r\ge0$. The CM of the resulting Gaussian state is of the form 
    \begin{eqnarray*}
    \V_{AB}=
    \begin{pmatrix}
    v\cosh{2r} & 0 & v\sinh{2r} & 0\\
    0 & v\cosh{2r} & 0 & -v\sinh{2r}\\
    v\sinh{2r} & 0 & v\cosh{2r} & 0\\
    0 & -v\sinh{2r} & 0 & v\cosh{2r}
    \end{pmatrix}.
    \end{eqnarray*}    
   Such state are entangled if and only if  $r>r_{ent}:=\frac{1}{2}\ln v$, following Ref. \cite{peres96ppt, horodecki96ppt, werner01bound, serafini03symplectic}. Meanwhile, we show in \cref{app:examples} that these states are Gaussian-atemporal if and only if
    $ r > r_{atemp}:= \frac{1}{2}\arccos{\frac{v+\sqrt{v^2+8}}{4}} > r_{ent}$. Therefore, while all Gaussian-atemporal states are entangled here, states for which $r_{ent} < r \le r_{atemp}$ are entangled and yet can still be explained temporally by a Gaussian channel between $A$ and $B$ (or vice versa, given its symmetry). However, when entanglement becomes sufficiently large (exceeding  $E_{max}:\approx 0.1882$ as measured by logarithmic negativity, which was introduced in~\cite{vidal02negativity}, and later proved to be an entanglement monotone in~\cite{plenio2005logarithmic}), Gaussian-atemporality is guaranteed \footnote{We use the natural logarithm for the logarithmic negativity throughout the paper.}.

\begin{figure*}[tbh]
\centering
        \subfloat[Pure states\label{subfig:pure}]{\includegraphics[width=0.35\textwidth]{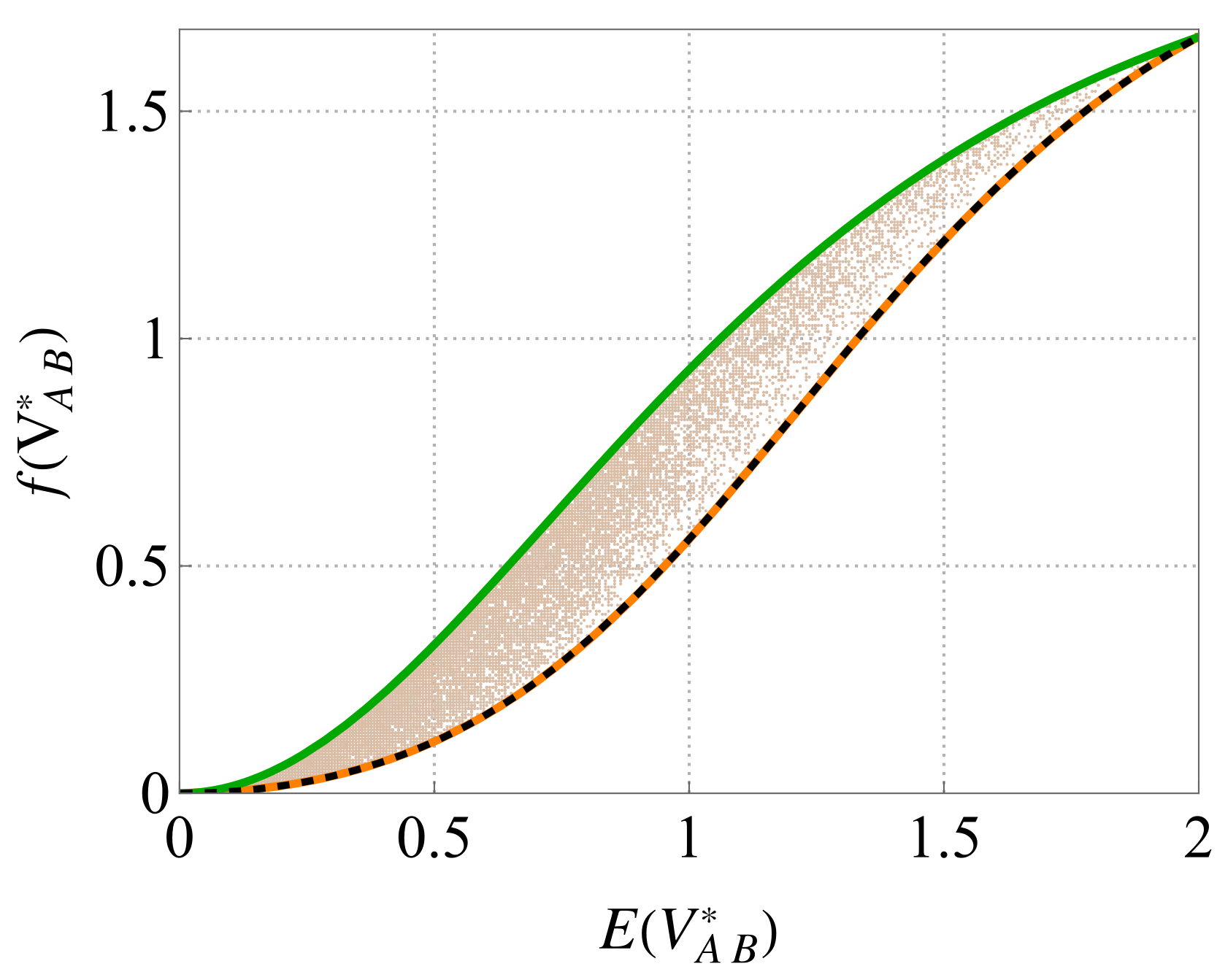}}
        \subfloat[Mixed states\label{subfig:mix}]{\includegraphics[width=0.35\textwidth]{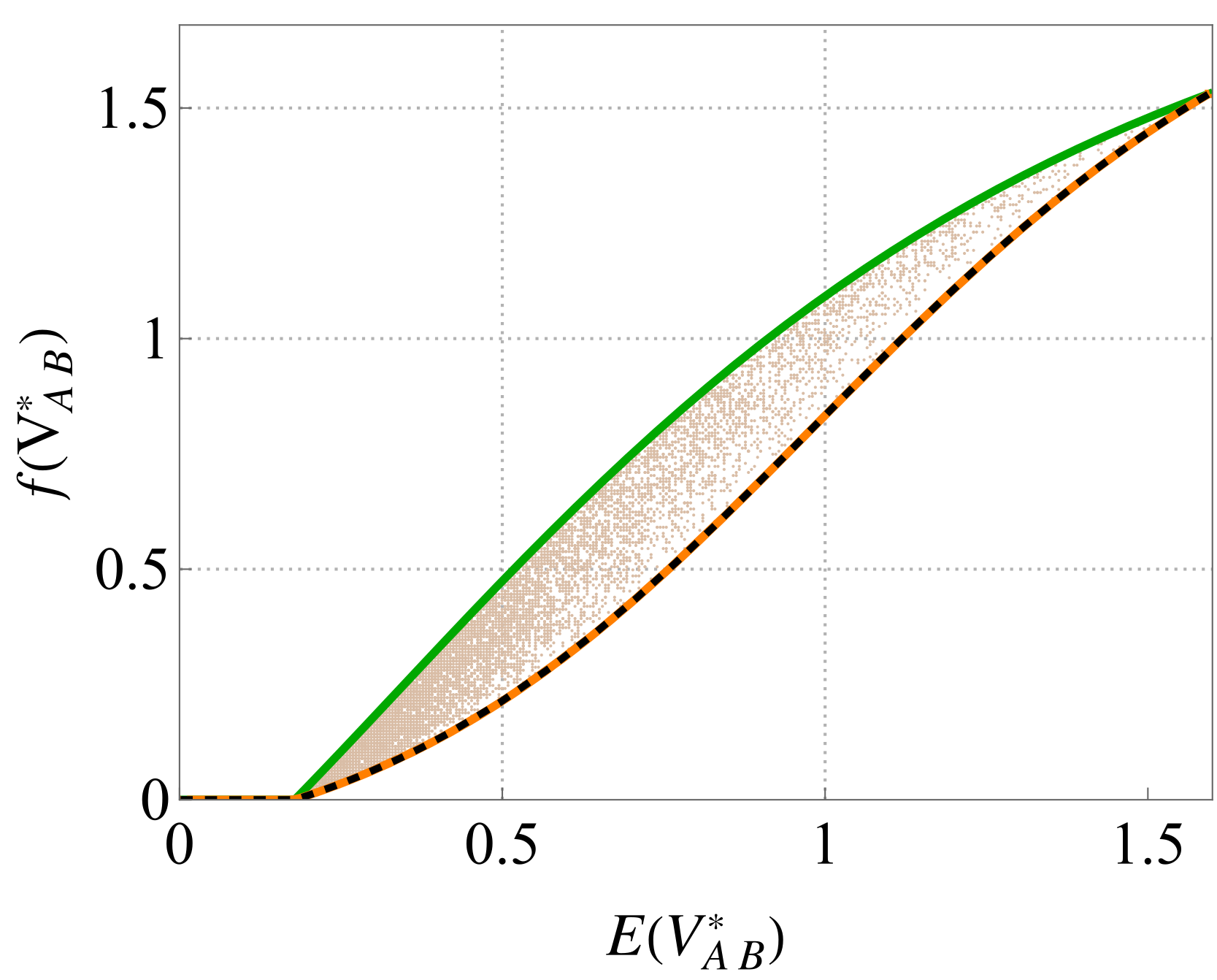}}
    \caption{\textbf{Entanglement and atemporality for randomly generated Gaussian states}. Consider randomly generated two-mode squeezed states. We use the logarithmic negativity~\cite{vidal02negativity,plenio2005logarithmic} as a measure of entanglement. (a) Plot of Gaussian-atemporality vs entanglement for $5000$ randomly generated pure state (with a squeezing $0 < r <1$). The top green line is the line for a two-mode squeezed vacuum state (with $0<r<1$ and $v=1$ ). The black dashed line is for a state we get from interfering two squeezed states (with $r=1$) while rotating the other. We also get the same behaviour (orange line) by interfering the two squeezed states on an asymmetric beam splitter. (b) Plot of atemporality vs entanglement for $20000$ randomly generated mixed states with the same level of mixedness. There is a finite value of the logarithmic negativity needed before we can observe any Gaussian-atemporality. The top green line is for a two-mode squeezed thermal state of $v=1.5$. The black dashed line is for a state that we get from interfering two squeezed thermal states while rotating one state relative to the other and the orange line is from interfering two squeezed thermal states at an asymmetric beam splitter.
    } 
    \label{fig:ent}
\end{figure*}

As a second example, consider a class of asymmetric states that described, two Gaussian states $A,B$, initially with local covariances $\V_A= (\begin{smallmatrix}
    u & 0\\
    0 & v
\end{smallmatrix})$, $\V_B= (\begin{smallmatrix}
    v & 0\\
    0 & u
\end{smallmatrix})$ that interact via a balanced beam splitter, where $uv\ge1$. Then, the resulting state has CM
\begin{eqnarray*}
    \V_{AB} = \frac{1}{2}\begin{pmatrix}
    u+v & 0 & -(u-v) & 0\\
    0 & u+v & 0 & u-v\\
    -(u-v) & 0 & u+v & 0\\
    0 & u-v & 0 & u+v
\end{pmatrix}.
\end{eqnarray*} 
The positive partial transpose (PPT) criterion \cite{peres96ppt, horodecki96ppt} tells us that the state is entangled if and only if $\min(u,v)<1$, i.e., our initial Gaussian states are squeezed and exhibit sub-vacuum noise. By using \gar, we find in \cref{app:examples} that it is Gaussian-atemporal if and only if 
\begin{eqnarray*}
    \frac{1-u}{u^2}>\frac{v-1}{v^2}.
\end{eqnarray*}
\cref{fig:ex_sym} then illustrates the relation between \gat\ and entanglement for various values of $u$ and $v$. In this case, all separable states are temporally compatible, and all atemporal states are entangled. This is indeed true in the case of qubit systems \cite{song2024class,song2025bipartite}. However, there is a region (shaded blue) that is entangled but still exhibits zero atemporality. These lead us to the following observation:

\begin{obs}
     There are entangled but temporally compatible Gaussian states, but sufficiently strong entanglement implies \gat.
\end{obs}

This observation tells us that entanglement and atemporality are distinct measures of non-classicality, with the latter likely being a stronger notion than the former; and mirrors similar observations for spatiotemporal states of two qubits. Indeed, Fig \ref{fig:ent} reaffirms this observation by visualising the relation between \gat\ and entanglement of randomly generated two-mode Gaussian states. From \cref{fig:ent}, \gat\ appears to imply entanglement. The converse, on the other hand is not generally true. There exist entangled and temporal Gaussian states. 

\subsection{Time-reversal Asymmetry of Atemporality}
When we consider more general spatio-temporal two-mode states, the property of time-reversal symmetry is of natural interest. In classical statistics, we know it is impossible to decide conclusively whether $A$ caused $B$ or $B$ caused $A$. Yet, studies on qubits have shown this is not necessarily true in the quantum regime. Gaussian-atemporality enables us to demonstrate this phenomenon in Gaussian systems.

Consider the scenario where Alice begins with a general Gaussian state whose covariance matrix has quadrature variances $\langle \hat{q}^2 \rangle = v_1$ and $\langle \hat{p}^2 \rangle = v_2$ when in standard form. Alice then sends the state to Bob via a loss channel with transmission rate $\eta$ (i.e., the associated channel has transformation matrix $\T=\sqrt{\eta}\I$ and noise matrix $\N = (1-\eta)\I$). The associated covariance matrix between Alice and Bob is then given by:
    \begin{eqnarray*}
        \V_{AB} = \begin{pmatrix}
    v_1 & 0 & v_1\sqrt{\eta} & 0\\
    0 & v_2 & 0 & v_2\sqrt{\eta}\\
    v_1\sqrt{\eta} & 0 & 1+(v_1-1)\eta & 0\\
    0 & v_2\sqrt{\eta} & 0 & 1+(v_2-1)\eta
\end{pmatrix}.    
    \end{eqnarray*}
    Clearly this state has zero forward atemporality, since it can be explained by a physical channel from Alice to Bob by definition. However, a causal explanation of this correlations as a channel from Bob to Alice is not always possible.

    In \cref{app:examples}, we reverse the role of $A$ and $B$ and use \cref{thm:cv_retrieve} to find the reverse atemporality robustness. Introducing $$v^{(\eta)}_k = \frac{v_k}{1+(v_k-1)\eta}$$ as the rescaled quadrature variances, we find
    \begin{eqnarray*}
        \overleftarrow{f}(\V_{AB}) &=& \max(0, (\eta v+1)(1-v)),
    \end{eqnarray*}
    where $v = \sqrt{v^{(\eta)}_1v^{(\eta)}_2}$. Noting that $v > 0$, $\overleftarrow{f}(\V_{AB})$ is then non-zero if and only if $v^{(\eta)}_1v^{(\eta)}_2 < 1$. From this, we observe the following:
    \begin{itemize}
        \item When there is no loss ($\eta = 1$), $v^{(\eta)}_k = 1$ for each $k$ and thus $v = 1$. This aligns with our expectation that any unitary channel has a valid time-reversal.
        \item When all input is lost ($\eta = 0$),  $v^{(\eta)}_k \rightarrow v_k$. Thus, zero reverse atemporality is guaranteed by the uncertainty principle, which implies $v_1v_2 \geq 1$.
        \item Whenever the initial state is classical (i.e., exhibits no squeezing such that $v_1, v_2 \geq 1$), then $v^{(\eta)}_1, v^{(\eta)}_2 \geq 1$ and thus $v \geq 1$ for any $\eta$. This aligns with our expectation that atemporality is a quantum phenomenon.
    \end{itemize}

    However, if our input is non-classical (i.e. squeezed), we can obtain non-zero reverse atemporality in moderate loss regimes. Fig.~\ref{fig:cv_ex1} illustrates this for various values of $\eta$. 
    We see that in regimes where the input is squeezed, there exist situations where the resulting spatio-temporal state can only be explained by direct cause from $A$ to $B$, but admits no similar explanation from $B$ to $A$. We thus conclude:

\begin{obs}
    \gat\ is asymmetric under time reversal. There exist situations where we can decide from correlation matrices alone whether $A$ caused $B$ or $B$ caused $A$.
\end{obs}

    \begin{figure}[ht]
    \centering
    \includegraphics[width=0.7\linewidth]{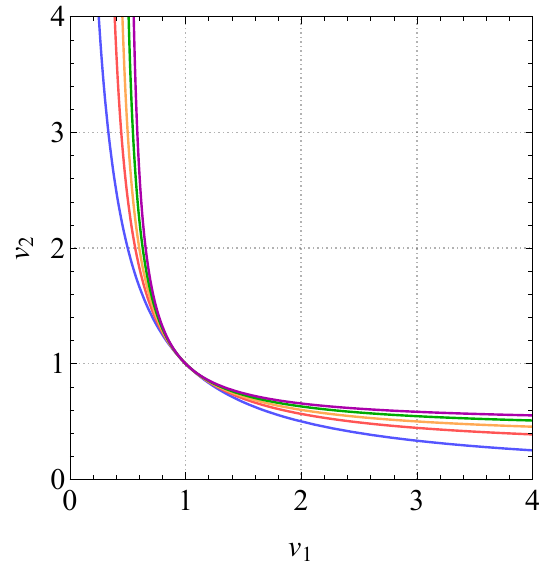}
    \caption{\textbf{Time-reversal asymmetry of atemporality.} The states constructed are parameterized by  $v_1,v_2$. By construction, these states are not forward Gaussian-atemporal. The region above the blue line represents physical space-time Gaussian states. The region below the red (yellow, green, purple) line  represents the ones that have non-zero reverse Gaussian-atemporalites, when $\eta=0.3$ ($\eta=0.5, 0.7, 0.9$, respectively). The region between these lines thus includes physical space-time Gaussian states that are not forward Gaussian-atemporal, but are reverse Gaussian-atemporal. }
    \label{fig:cv_ex1}
\end{figure}

\section{Discussion}
Given two random variables drawn from some multi-variable Gaussian distribution, Gaussian-atemporality asks: can we rule out the resulting correlations between them from being generated purely by temporal mechanisms, such that they correspond to two measurements on the same physical system? Here, we answer this in the affirmative, illustrating scenarios where the covariance matrix between quadrature measurements could not arise if performed on the same quantum harmonic oscillator at two different times, regardless of what Gaussian channel occurs in between. We introduced \gar -- the maximal noise before breaking \gat -- as an operationally meaningful measure of \gat, and illustrated its expression in closed analytical form. Use of these tools then provides two immediate observations: (i) \gat\ provides a distinct notion of quantum correlations beyond entanglement and (ii) that certain covariances correlating quadrature measurements of Gaussian systems $A$ and $B$ have a natural arrow of time - they can be explained when $A$ comes before $B$ but not vice versa.

Our work here complements recent studies of spatiotemporal correlations on qubits, where the use of pseudo-density operators enables the discovery that specific correlations are temporal and may possess an intrinsic arrow of time. Our work thus lends further credence that such phenomena extend well beyond quantum systems of few dimensions. Meanwhile, pseudo-density operators have recently motivated spatiotemporal generalizations of quantum mutual information~\cite{wu2025mutual} and teleportation~\cite{marletto21temporaltele} and seen extensions to multi-time states \cite{jia2024double,liu2024unification} or quantum state over time~\cite{leifer13causalbayesian,horsman17pdoandothers,fullwood22onquantumstates,lie2024unique,lie2024uniquemulti,lie2025interferometry}. Similarly, pseudo-channels have conceptual similarity with Bayesian retrodiction maps or Bayesian inverse \cite{aw2021retrodiction,parzygnat22binverse,parzygnat2023axioms} which include Petz recovery maps~\cite{petz1986sufficient,petz1988sufficiency}. There would certainly be value in extending these ideas to the continuous-variable Gaussian domain. Indeed, ever since the initial proposal of the EPR paradox, Gaussian quantum states have been a natural testbed for uniquely quantum phenomena - thanks to their concise covariance matrix description and relative ease of experimental realization. Moreover, the Gaussian continuous variables enjoy the advantage over discrete variables of having an efficiently computable measure of atemporality. Thus, we anticipate they can provide a compelling pathway towards understanding general multi-partite spatiotemporal quantum correlations.

\section*{Acknowledgements} 
This project began while MS was in Nanyang Quantum Hub, School of Physical and Mathematical Sciences, Nanyang Technological University. This work is supported by the National Research Foundation of Singapore through the NRF Investigatorship Program (Award No. NRF-NRFI09-0010), the National Quantum Office, hosted in A*STAR, under its Centre for Quantum Technologies Funding Initiative
(S24Q2d0009), grant FQXi-RFP-IPW-1903 (``Are quantum agents more energetically efficient at making predictions?") from the Foundational Questions Institute (FQxI) and Fetzer Franklin Fund (a donor-advised fund
of Silicon Valley Community Foundation), the Singapore Ministry of Education Tier 1 Grant RG77/22 and the Australian Research Council Centre of Excellence for Quantum Computation and Communication Technology (Project No. CE110001027). 
 
\appendix \label{app}

\section{Proofs of technical results.} \label{app:proofs}

Recall that measurement statistics are given in the standard basis $\ax^\star$ such that the mean vectors are zero and local covariance matrices (CMs) are diagonal. Henceforth, we assume the standard basis unless otherwise specified. For brevity, we will often omit $^\star$ symbol. We are concerned with Gaussian channels with no displacement effect throughout the paper, thus we only need two parameters $\T$ and $\N$ to characterize a Gaussian channel. 

\setcounter{lem}{0}
\begin{lem}
    Consider a temporal causal structure where an initial Gaussian state, characterized by a CM $\V_A$ evolves in time through a Gaussian channel $\G:A\to B$. Let $\T$ be a transformation matrix associated with $\G$. Then, the cross-correlation matrix $\C$ of the space-time CM,  i.e., $\V_{AB}\equiv ( \begin{smallmatrix}
\V_A & \C \\
{\C}\t & \V_B
\end{smallmatrix} )$, is given by 
     \begin{eqnarray}
        \C = \V_A\T\t. 
     \end{eqnarray} 
\end{lem}

\begin{proof}
Since $\bx_A, \V_{A}$ are in the standard form, $\bx_A=\zero$ and an element of $\V_{A}$ is of the form $V_{A,jk} = v_j \delta_{jk}$, where $v_{1}\equiv \expval{\hq_A^2}, v_{2}\equiv \expval{\hp_A^2}$ and $\delta_{jk}$ represents the Kronecker delta. To prove $\C = \V_A\T^t$, we will show that 
    \begin{eqnarray}
        C_{ij} = v_i \sum_k \delta_{ik} T_{jk}= v_i T_{ji}.
    \end{eqnarray}

We are concerned with a temporal causal mechanism by which a initial state and a Gaussian channel are given by $\rho_A$ and $\G:A \to B$, respectively. Following Ref.~\cite{zhang20different}, the cross-correlation matrix 
$C_{ij} \equiv \frac{1}{2}\expval{\hx_{A,i} \hx_{B,j} + \hx_{B,j}\hx_{A,i}}$ of the corresponding temporally compatible space-time CM is given by  
\begin{eqnarray}
     C_{ij} &=& \int \dd{x_A} x_A  \tr\left[ \G\left(\widehat{\Pi}^{(i)}_{x_A}\rho_A\widehat{\Pi}^{(i)}_{x_A} \right) \hx_j^{(B)} \right] \\
    &=& \int  \dd{x_A} x_A  \tr\left[\widehat{\Pi}^{(i)}_{x_A}\rho_A\widehat{\Pi}^{(i)}_{x_A} \G\+(\hx_j^{(B)}) \right], \label{eq:cv_corr}
\end{eqnarray}
where $\widehat{\Pi}^{(i)}_{x_A}$ represents the projector of $\hx_{A,i}$ associated with its eigenvalue $x_A$, and  $\widehat{\Pi}^{(i)}_{x_A}\rho_A\widehat{\Pi}^{(i)}_{x_A}$ represents the unnormalized post-measurement state after achieving the outcome $x_A$ when measured with $\hx_{A,i}$. Also, $\G\+$ represents Hilbert--Schmidt adjoint of $\G$. Before we complete the proof, let us briefly take a look at how $\T,\N,\d$ play roles in transforming quadrature operators below.

The Hilbert--Schmidt adjoint map $\G\+$ transforms the quadrature operators on $B$ into those on $A$. Without loss of generality, $\hx_{A,i}$ and $\hx_{B,j}$ can be related by introducing quadrature operators $\ax_E$ of an uncorrelated environmental system in a state $\rho_E$ as
\begin{eqnarray}
    \G\+:\hx_{B,j} \to \G\+(\hx_{B,j})=\sum_k T_{jk}\hx_{A,k} + \hx_{E,j}, \label{eq:cv_gadj}
\end{eqnarray}
with some linear transformation matrix $\T$. The local CM, $\V_B$, on Bob's side then becomes 
\begin{widetext}
    \begin{eqnarray}
    V_{B,jk} &\equiv& 
    \frac{1}{2}\expval{\Delta\hx_{B,j} \Delta\hx_{B,k} + \Delta\hx_{B,k}\Delta\hx_{B,j}}_{\G(\rho_A)} \\
    &=& \frac{1}{2}\expval{\Delta\G\+(\hx_{B,j}) \Delta\G\+(\hx_{B,k}) + \Delta\G\+(\hx_{B,k})\G\+(\Delta\hx_{B,j})}_{\rho_A\rho_E} \\
    &=& \sum_{l}\sum_{m} T_{jl}\tr\left[ \frac{1}{2}\left(\Delta\hx_{A,l}\Delta\hx_{A,m}+\Delta\hx_{A,m}\Delta\hx_{A,l}\right)\rho_A\right]T_{km} + \tr\left[ \frac{1}{2}\left(\Delta\hx_{E,j}\Delta\hx_{E,k}+\Delta\hx_{E,k}\Delta\hx_{E,j}\right)\rho_E\right] \nonumber\\
    &+& \tr\left[\left(\sum_{l} T_{jl}\hx_{A,l}\hx_{E,k}\right)\rho_A\rho_E\right]+\tr\left[\left(\sum_{m} T_{km}\hx_{A,m}\hx_{E,j}\right)\rho_A\rho_E\right]\\
    &=& \sum_{l,m} T_{jl}\left(\frac{1}{2}\expval{\Delta\hx_{A,l}\Delta\hx_{A,m}+\Delta\hx_{A,m}\Delta\hx_{A,l}}_{\rho_A}\right)T_{km} + \frac{1}{2}\expval{\Delta\hx_{E,j}\Delta\hx_{E,k}+\Delta\hx_{E,k}\Delta\hx_{E,j}}_{\rho_E} \nonumber\\
    &+& \sum_{l} T_{jl}\tr\left[\hx_{A,l}\rho_A\right]\tr\left[\hx_{E,k}\rho_E\right]+\sum_{m}T_{km}\tr\left[\hx_{A,m}\rho_A\right]\tr\left[\hx_{E,j}\rho_E\right] \\
    &=& \sum_{l,m} T_{jl} V_{A,lm} T_{km} + N_{jk},
\end{eqnarray}
\end{widetext}
where $N_{jk}\equiv\frac{1}{2}\expval{\{\Delta\hx_{E,j},\Delta\hx_{E,k}\}}_{\rho_E} + \sum_{l} T_{jl}\tr[\hx_{A,l}\rho_A]\tr[\hx_{E,k}\rho_E]+\sum_{m}T_{km}\tr[\hx_{A,m}\rho_A]\tr[\hx_{E,j}\rho_E]$. The subscript of the expectation value symbol $\expval{\cdot}$ represents the quantum state that is being measured. In the third equation, we used Eq.~(\ref{eq:cv_gadj}) and the fact that the system of interest and the environmental system are uncorrelated. Here we omit tensor products for brevity, but note that operators with subscript $_A$ and those with $_E$ live in different Hilbert spaces, thus they commute. Consequently, the fourth equation follows. 

Therefore, we have the relation $\V_B = \T\V_A\T\t +\N$. These matrix $\T$ and $\N$ fully characterize the Gaussian channel $\G$ together with a displacement vector $\d\equiv( \begin{smallmatrix}
\expval{\hq_E} \\ \expval{\hp_E}
\end{smallmatrix})$ \cite{weedbrook12gaussian, holevo2001evaluating}.

Returning to Eq.~(\ref{eq:cv_corr}), we substitute $\sum_k T_{jk}\hx_{A,k} + \hx_{E,j}$ for $\G\+(\hx_{B,j})$ as in the box above, but with $\expval{\hx_{E,j}} =0$ as we are concerned with CMs in standard basis:
\begin{widetext}
    \begin{eqnarray}
    & &C_{ij} = \int \dd{x_A} x_A  \tr\left[\widehat{\Pi}^{(i)}_{x_A}\rho_A\widehat{\Pi}^{(i)}_{x_A} \G\+(\hx_j^{(B)}) \right]\\
    &=& \sum_k T_{jk}\int \dd{x_A} x_A  \tr\left[\widehat{\Pi}^{(i)}_{x_A}\rho_A\widehat{\Pi}^{(i)}_{x_A} \hx_{A,k} \right] + \int \dd{x_A} x_A  \tr\left[\widehat{\Pi}^{(i)}_{x_A}\rho_A\widehat{\Pi}^{(i)}_{x_A} \hx_{E,j} \right]\\
    &=& T_{j,k=i} \int \dd{x_A} x_A  \tr\left[\widehat{\Pi}^{(i)}_{x_A}\rho_A\widehat{\Pi}^{(i)}_{x_A} \hx_{A,k=i} \right] + T_{j,k\ne i} \int \dd{x_A} x_A  \tr\left[\widehat{\Pi}^{(i)}_{x_A}\rho_A\widehat{\Pi}^{(i)}_{x_A} \hx_{A,k\ne i} \right] \\
    &=& T_{ji} \int \dd{x_A} x_A  \tr\left[\widehat{\Pi}^{(i)}_{x_A}\rho_A\widehat{\Pi}^{(i)}_{x_A} \left(\int \dd{x'_A} x'_A \widehat{\Pi}^{(i)}_{x'_A} \right)\right] + T_{j,k\ne i} \int \dd{x_A} x_A  \tr\left[\widehat{\Pi}^{(i)}_{x_A}\rho_A\widehat{\Pi}^{(i)}_{x_A} \left( \int \dd{x'_A} x'_A \widehat{\Pi}^{(k\ne i)}_{x'_A} \right) \right] \\
    &=& T_{ji} \int\int \dd{x_A} \dd{x'_A} x_A x'_A \delta_{x_A,x'_A} \tr\left[\widehat{\Pi}^{(i)}_{x_A}\rho_A \right] + T_{j,k\ne i} \int \int \dd{x_A} \dd{x'_A} x_A x'_A \tr\left[\rho_A\Pi^{(i)}_{x_A} \widehat{\Pi}^{(k\ne i)}_{x'_A} \widehat{\Pi}^{(i)}_{x_A} \right] \\
    &=& T_{ji} \expval{\hx_{A,i}^2} + T_{j,k\ne i} \int  \dd{x_A} x_A \left( \underbrace{\int \dd{x'_A}  x'_A}_{=0} \right) \tr\left[\rho_A c_{x_A,x'_A}\widehat{\Pi}^{(i)}_{x_A} \right]  \\
    &=& v_i T_{ji}.
\end{eqnarray}
\end{widetext}
In the third equation, we used $\int \dd{x_A} x_A  \tr\left[\widehat{\Pi}^{(i)}_{x_A}\rho_A\widehat{\Pi}^{(i)}_{x_A} \hx_{E,j} \right]=0$ because $\expval{\hx_{E,j}} =0$ and $\hx_{E,j}$ is an operator on $E$ not on $A$. We also used $\Pi^{(i)}_{x_A} \widehat{\Pi}^{(i)}_{x'_A}=\delta_{x_A,x'_A}\widehat{\Pi}^{(i)}_{x_A}$ and $\widehat{\Pi}^{(i)}_{x_A} \widehat{\Pi}^{(k\ne i)}_{x'_A} \widehat{\Pi}^{(i)}_{x_A} = c_{x_A,x'_A}\widehat{\Pi}^{(i)}_{x_A}$ for some constant $c_{x_A,x'_A}$ as $\widehat{\Pi}^{(i)}_{x_A}$ is a rank-1 projector. Consequently, we have $C_{ij}=v_i T_{ji}$. \\

Hence, we conclude that $\C = \V_A\T\t$.
\end{proof}
 
\setcounter{thm}{2}
\begin{thm}
    A spatial-temporal qumode-pair with covariance matrix $\V_{AB}$ and an associated forward pseudo-channel with transformation matrix $\T$ and noise matrix $\N$ has forward atemporality robustness
    \begin{eqnarray}
        \overrightarrow{f}(\V_{AB})=\max(0, \abs{\omega}-\sqrt{\det\N}),
    \end{eqnarray}
    with $\omega = 1-\det\T=1-\frac{\det\C}{\det\V_A}$ and $\det\N=\frac{\det\V_{AB}}{\det \V_A}$. Moreover, $\overrightarrow{f}(\V_{AB}) > 0$ is a necessary and sufficient condition for forward Gaussian atemporality.   
\end{thm}

\begin{proof}
    Here, we prove 
    \begin{align}    \sup_{\substack{\mathbf{E}\ge \zero \text{ s.t.}\\ \X(\T,\N+\mathbf{E})< \zero}} \mathcal{N}(\mathbf{E}) \quad = \quad  \max(0,\abs{\omega}-\sqrt{\det\mathbf{N}}), 
    \end{align}
    when $\mathcal{N}(\mathbf{E})\equiv \sqrt{\det \mathbf{E}}$ is used as a measure of noise and $\omega=1-\frac{\det\C}{\det\V_A}$.

    Let us take a quick look at properties of a forward atemporality matrix. Given a space-time covariance matrix $\V_{AB}$, the forward pseudo-channel (characterized by its transformation matrix $\T$ and noise matrix $\N\ge\zero$) is uniquely determined, consequently the forward atemporality matrix $\X$ is also uniquely determined as
\begin{eqnarray}
    \X(\T,\N) = \N+i\M, \quad \with \quad \M\equiv \bOmega-\T\bOmega \T\t.
\end{eqnarray}
Using $\T\bOmega \T\t=(\det\T) \bOmega$, we note that 
\begin{align}
    \M &\equiv \bOmega-\T\bOmega \T\t\\
    &= (1-\det\T)\bOmega\\
    &= \left(1-\frac{\det\C}{\det\V_A}\right) \bOmega  =: \omega\bOmega,
\end{align}
where in the third equation we used \cref{thm:cv_retrieve}, i.e., $\T=\C\t\V_A^{-1}$. Thus we can write $\X(\T,\N)$ as $\X(\T,\N) = \N+i\omega\bOmega$. In turn, we observe that $\X(\T,\N)\ge \zero \Leftrightarrow \abs{\omega}-\sqrt{n_1n_2}\le 0$, where $n_1,n_2$ are the eigenvalues of $\N$:

\begin{align}
        \X(\T,\N)= \N + i\omega\bOmega \ge \zero &\Leftrightarrow \det\X(\T,\N) \ge \zero \\
&\Leftrightarrow \det\mathbf{O} \X(\T,\N) \mathbf{O}\t \ge \zero\\
&\Leftrightarrow \det \begin{pmatrix}
        n_1 & i\omega \\ -i\omega & n_2
\end{pmatrix}\ge \zero\\
&\Leftrightarrow \abs{\omega} - \sqrt{n_1n_2}\le 0,
    \end{align}
    where $\mathbf{O}$ is an orthogonal transformation that diagonalizes $\N$, i.e., $\mathbf{O}\N\mathbf{O}\t = \begin{pmatrix}
        n_1 & 0 \\ 0 & n_2
\end{pmatrix}$, where $n_1,n_2$ are the eigenvalues of $\N$. The first line follows from $\tr \X(\T,\N) \ge 0$ because 
$\tr\N \ge 0$ and $\tr\bOmega=0$. The second line follows from $\det\mathbf{O} =1$. The third line follows from $\mathbf{O}\bOmega\mathbf{O}\t = \bOmega$.

Now, (1) let us consider an diagonal noise matrix to be added to $\N$, thus resulting in a new noise matrix $\tilde{\N} = \N+\mathbf{E}_\text{diag}$, where $\mathbf{E}_\text{diag}:=\begin{pmatrix}
    \epsilon_1 & 0\\
    0 & \epsilon_2
\end{pmatrix}$ represents the additional noise matrix. In this case, the new noise matrix has its eigenvalues as $\tilde{n}_1=n_1+\epsilon_1$ and $\tilde{n}_2=n_2+\epsilon_2$. Provided that, we will show that 
\begin{eqnarray}
     \sup_{\abs{w}-\sqrt{\tilde{n}_1\tilde{n}_2}>0} \mathcal{N}(\mathbf{E}_\text{diag}) = \abs{w}-\sqrt{n_1n_2},
\end{eqnarray}
when $\abs{w}-\sqrt{n_1n_2}>0$. Or equivalently, 
 \begin{eqnarray}
     \max_{\abs{w}-\sqrt{(n_1+\epsilon_1)(n_2+\epsilon_2)}=0} \sqrt{\epsilon_1\epsilon_2} = \abs{w}-\sqrt{n_1n_2}.
\end{eqnarray}
The maximization problem on the left hand side can be solved by using the method of Lagrangian multiplier as follows: Define a Lagrangian function as 
\begin{eqnarray}
    \mathcal{L}(\epsilon_1,\epsilon_2,\lambda):= \epsilon_1\epsilon_2 + \lambda(w^2-(n_1+\epsilon_1)(n_2+\epsilon_2)).
\end{eqnarray}
Then, $\epsilon_1\epsilon_2$ (and also $\sqrt{\epsilon_1\epsilon_2}$) takes its maximum when $\frac{\partial\mathcal{L}}{\partial \epsilon_i}=0, \forall~i$ and $\frac{\partial\mathcal{L}}{\partial \lambda}=0$. That is, 
\begin{eqnarray}
    \frac{\partial\mathcal{L}}{\partial \epsilon_1}&=&0 : \epsilon_2 = -n_2 +\frac{n_2}{1-\lambda}\\
    \frac{\partial\mathcal{L}}{\partial \epsilon_2}&=&0 : \epsilon_1 = -n_1 +\frac{n_1}{1-\lambda}\\
    \frac{\partial\mathcal{L}}{\partial \lambda}&=&0 : w^2=(n_1+\epsilon_1)(n_2+\epsilon_2).
\end{eqnarray}
After solving these simultaneous equations in terms of $\epsilon_1,\epsilon_2,\lambda$, we have $\lambda = 1 \pm \frac{\sqrt{n_1n_2}}{\abs{w}}$. But $\lambda = 1 + \frac{\sqrt{n_1n_2}}{\abs{w}}$ leads to negative $\epsilon$'s, which contradicts to the fact that a noise matrix is positive semidefinite. Thus, the solution becomes
\begin{eqnarray}
    \epsilon_1 &=& \sqrt{\frac{n_1}{n_2}}(\abs{w}-\sqrt{n_1n_2}), \\
    \epsilon_2 &=& \sqrt{\frac{n_2}{n_1}}(\abs{w}-\sqrt{n_1n_2}), \\
    \lambda &=& 1 - \frac{\sqrt{n_1n_2}}{\abs{w}}.
\end{eqnarray}
Then it simply follows that the maximum of $\sqrt{\epsilon_1\epsilon_2}$ is given by 
\begin{eqnarray}
   &&\max_{\abs{w}-\sqrt{(n_1+\epsilon_1)(n_2+\epsilon_2)}=0} \sqrt{\epsilon_1\epsilon_2} \\
   &=& \sqrt{\sqrt{\frac{n_1}{n_2}}(\abs{w}-\sqrt{n_1n_2})\sqrt{\frac{n_2}{n_1}}(\abs{w}-\sqrt{n_1n_2})}\\
    &=& \abs{w}-\sqrt{n_1n_2}.
\end{eqnarray}

(2) Now we consider a more general case where the additional noise matrix is not diagonal, i.e., $\mathbf{E}:=\begin{pmatrix}
    \epsilon_1 & \epsilon_3\\
    \epsilon_3 & \epsilon_2
\end{pmatrix}$. Let $\tilde{\N}$ be the new noise matrix after adding $\mathbf{E}$ to $\N$, i.e., $\tilde{\N}=\mathbf{E}+\N$, and let $\tilde{n}_1,\tilde{n}_2$ be the eigenvalues of $\tilde{\N}$. Similar to the case (1) above, we will show that 
\begin{eqnarray}
    \abs{w}-\sqrt{n_1n_2} = \max_{\abs{w}-\sqrt{\tilde{n}_1\tilde{n}_2}=0} \sqrt{\det \mathbf{E}}.
\end{eqnarray}
We then define a Lagrangian function as 
\begin{eqnarray}
     &&\mathcal{L}(\epsilon_1,\epsilon_2,\epsilon_3,\lambda)\\
     &:=& \det \mathbf{E} + \lambda(w^2-(n_1+\epsilon_1)(n_2+\epsilon_2)+\epsilon_3^2)\\
     &=& \epsilon_1\epsilon_2-\epsilon_3^2 + \lambda(w^2-(n_1+\epsilon_1)(n_2+\epsilon_2)+\epsilon_3^2).
\end{eqnarray}
When $\frac{\partial\mathcal{L}}{\partial \epsilon_i}=0, \forall~i$ and $\frac{\partial\mathcal{L}}{\partial \lambda}=0$, $\epsilon_1\epsilon_2-\epsilon_3^2$ is maximized. In particular, $\frac{\partial\mathcal{L}}{\partial \epsilon_3}=0$ leads to $(\lambda+1)\epsilon_3=0$. If $\lambda=-1$, it leads to a solution which contradicts to $\mathbf{E}\ge \zero$. Thus, we have $\epsilon_3=0$. Then the maximization problem becomes the same as the one considered in the case (1), that is, it follows that 
\begin{eqnarray}
    \max_{\abs{w}-\sqrt{\tilde{n}_1\tilde{n}_2}=0} \sqrt{\det \mathbf{E}} = \abs{w}-\sqrt{n_1n_2}.
\end{eqnarray}

Hence, we conclude that \gat\ amounts to the maximum amount of noise that can be added before breaking \gat.

Lastly, we verify $\det\N=\frac{\det\V_{AB}}{\det \V_A}$. Note that $\N$ coincides with the Schur complement of block matrix $\V_A$ of the whole matrix $\V_{AB}$, denoted by $\V_{AB}/\V_A$ \cite{Schur1918schurcomplement,zhang2006schur}. Following Schur determinant formula (see Eq.~(0.3.2) in Ref.~\cite{zhang2006schur}), the determinant of $\N=\V_{AB}/\V_A$ is easily computed by $\frac{\det\V_{AB}}{\det \V_A}$.  
\end{proof}

\section{Examples.}\label{app:examples}

\textit{Example 1.} Consider two thermal states with variance $v\ge1$, each labeled by $A$ and $B$, together passing through a two-mode squeezing operation with a squeezing parameter $r\ge0$. The CM of the resulting Gaussian state is of the form 
    \begin{eqnarray*}
    \V_{AB}=
    \begin{pmatrix}
    v\cosh{2r} & 0 & v\sinh{2r} & 0\\
    0 & v\cosh{2r} & 0 & -v\sinh{2r}\\
    v\sinh{2r} & 0 & v\cosh{2r} & 0\\
    0 & -v\sinh{2r} & 0 & v\cosh{2r}
    \end{pmatrix}.
    \end{eqnarray*}    
    Then, the states are Gaussian-atemporal for 
    $ r > r_{atemp}:= \frac{1}{2}\arccos{\frac{v+\sqrt{v^2+8}}{4}}$, and entangled for $r>r_{ent}:=\frac{1}{2}\ln v$.
    In this class of Gaussian states, we observe that Gaussian-atemporal states are always entangled. However, there exist the states that are entangled but not Gaussian-atemporal, when $r_{ent} < r \le r_{atemp}$. Interestingly though, if the amount of entanglement, measured by the logarithmic negativity \cite{vidal02negativity,plenio2005logarithmic}, exceeds 
    $E_{max}:\approx 0.1882$, it guarantees that the state is Gaussian-atemporal. The logarithm is taken to be natural base for the logarithmic negativity.

\begin{proof}
    First of all, we note that $\overrightarrow{f}(\V_{AB})=\overleftarrow{f}(\V_{AB})=f(\V_{AB})$ since space-time CMs of the form as above are invariant under swap operations. Therefore, we will use \gar\ to derive the conditions for (1) \gat, and (2) use the positive partial transposition (PPT) \cite{peres96ppt, horodecki96ppt} to derive the conditions for entanglement. (3) We also prove that the logarithmic negativity of these states is maximized by
    \begin{eqnarray}
        \max_{f(\V_{AB})=0} E(\V_{AB}) \approx 0.1882,
    \end{eqnarray}
    where $E(\cdot)$ represents the logarithmic negativity \cite{vidal02negativity,plenio2005logarithmic} with natural base.\\

    (1) To derive the conditions for \gat, we first need to find the forward Gaussian pseudo-channel, characterized by $\T$ and $\N$. By using \cref{thm:cv_retrieve}, they are given by 
    \begin{eqnarray}
        \T = \tanh{2r}\mathbf{Z} \quad \text{and} \quad \N = \frac{v}{\cosh{2r}}\I,
    \end{eqnarray}
    where $\mathbf{Z}\equiv \begin{pmatrix}
    1 & 0 \\
    0 & -1 \\
\end{pmatrix}$ and $\I\equiv \begin{pmatrix}
    1 & 0 \\
    0 & 1 \\
\end{pmatrix}$. By using this, we can readily derive that 
\begin{eqnarray}
    \overrightarrow{f}(\V_{AB}) &=& \max(0, 1+\frac{\sinh^2{2r}}{\cosh^2{2r}}-\frac{v}{\cosh{2r}})\\
    &=& \max(0, \frac{2\cosh^2{2r}-v\cosh{2r}-1}{\cosh^2{2r}}), 
\end{eqnarray}
and $\overrightarrow{f}(\V_{AB})=0$ if and only if $2\cosh^2{2r}-v\cosh{2r}-1\le 0$; otherwise, it takes a non-zero value. Let $\zeta:=\cosh{2r}$ be a positive variable such that $\zeta\ge1$ and $h$ a function of $\zeta$, defined as $h(\zeta)\equiv 2\zeta^2-v\zeta-1$. That is, $\overrightarrow{f}(\V_{AB})=0$ if and only if $h(\zeta)\le 0$. Now let $\zeta^{(+)}$ be the largest solution of the quadratic equation $h(\zeta)=0$, i.e., $\zeta^{(+)}:=\frac{v+\sqrt{v^2+8}}{4}$. 
 
Since $v\ge1$, it simply follows that $\zeta^{(+)}\ge1$. Note also that the other solution $\zeta^{(-)}:=\frac{v-\sqrt{v^2+8}}{4}$ always takes a negative value, but this is not a feasible solution because $\zeta:=\cosh{2r}\ge1$. 

For $1 \le \zeta \le \zeta^{(+)}$, we have $h(\zeta)\le h(\zeta^{(+)}) = 0$, which implies $\overrightarrow{f}(\V_{AB})=0$; For $\zeta > \zeta^{(+)}$, on the other hand, we have $h(\zeta)>0$, which implies $\overrightarrow{f}(\V_{AB})>0$. Thus we arrive at the following:
\begin{itemize}
    \item It is Gaussian-atemporal for $\cosh{2r}>\frac{v+\sqrt{v^2+8}}{4}$.
\end{itemize}

(2) We use the PPT criterion to derive the conditions for entanglement \cite{peres96ppt, horodecki96ppt}, and the fact that the PPT criterion becomes a necessary and sufficient condition for entanglement for two-mode Gaussian states \cite{werner01bound}. After applying the partial transpose, the CM $\V_{AB}^\star$ changes to $\tilde{\V}_{AB}\equiv(\I\oplus \mathbf{Z})\V_{AB}^\star(\I\oplus\mathbf{Z})$. Let $\tilde{v}^{(-)}$ be the smallest value in the symplectic spectrum of $\tilde{\V}_{AB}$. Then, $\tilde{v}^{(-)}<1$ is a necessary and sufficient condition for entanglement \cite{serafini03symplectic}.

Now let us find $\tilde{v}^{(-)}$: $\tilde{v}^{(-)}$ is the smallest value in symplectic spectrum of
    \begin{eqnarray}
        \Tilde{\V}_{AB} &=& (\I\oplus \mathbf{Z})\V_{AB}(\I\oplus\mathbf{Z})\\
        &=& \begin{pmatrix}
    v\cosh{2r} & 0 & v\sinh{2r} & 0\\
    0 & v\cosh{2r} & 0 & v\sinh{2r}\\
    v\sinh{2r} & 0 & v\cosh{2r} & 0\\
    0 & v\sinh{2r} & 0 & v\cosh{2r}
\end{pmatrix} \\
&=:& \begin{pmatrix}
    \mathbf{A} & \mathbf{C} \\
    \mathbf{C} & \mathbf{B} 
\end{pmatrix},
    \end{eqnarray} and then it can be found by using the following known formula \cite{serafini2004eigensymplectic},
    \begin{eqnarray}
        \tilde{v}^{(-)} = \sqrt{\frac{\tilde{\Delta}-\sqrt{\tilde{\Delta}^2-4\det{\Tilde{\V}_{AB}}}}{2}},
    \end{eqnarray}
    where $\mathbf{A},\mathbf{B},
    \mathbf{C}$ are 2 by 2 matrices, and $\tilde{\Delta} \equiv \det\mathbf{A} + \det\mathbf{B} + 2\det\mathbf{C}$. In this case, $\Tilde{\Delta}$ amounts to 
    \begin{eqnarray}
        \tilde{\Delta} = 2v^2\cosh^2{2r}+2v^2\sinh^2{2r}.
    \end{eqnarray}
    Also, we can compute for $\det\Tilde{\V}_{AB}$:
    \begin{eqnarray}
        \det\Tilde{\V}_{AB} &=& v\cosh{2r} \begin{vmatrix}
           v\cosh{2r} & 0 & v\sinh{2r} \\
            0 & v\cosh{2r} & 0 \\
            v\sinh{2r} & 0 & v\cosh{2r}
        \end{vmatrix}\nonumber
        \\
        &&+v\sinh{2r} \begin{vmatrix}
            0 & v\cosh{2r} & v\sinh{2r}\\
            v\sinh{2r} & 0 & 0\\
            0 & v\sinh{2r} & v\cosh{2r}
        \end{vmatrix}\nonumber\\
        &=& v^2\cosh^2{2r}(v^2\cosh^2{2r}-v^2\sinh^2{2r}) \nonumber\\
        &&- v^2\sinh^2{2r} (v^2\cosh^2{2r} - v^2\sinh^2{2r})\nonumber\\
        &=& (v^2\cosh^2{2r}-v^2\sinh^2{2r})^2
    \end{eqnarray}
    
    Therefore, we have 
    \begin{widetext}
        \begin{eqnarray}
        \tilde{v}^{(-)} &=& \sqrt{\frac{\tilde{\Delta}-\sqrt{\tilde{\Delta}^2-4\det{\Tilde{\V}_{AB}}}}{2}}\\
        &=& \sqrt{v^2\cosh^2{2r}+v^2\sinh^2{2r}-\sqrt{(v^2\cosh^2{2r}+v^2\sinh^2{2r})^2-(v^2\cosh^2{2r}-v^2\sinh^2{2r})^2}}\\
        &=& \sqrt{v^2\cosh^2{2r}+v^2\sinh^2{2r}-\sqrt{4v^4\cosh^2{2r}\sinh^2{2r}}}\\
    &=& \sqrt{v^2\cosh^2{2r}+v^2\sinh^2{2r}-2v^2\cosh{2r}\sinh{2r}}\\
        &=& v(\cosh{2r}-\sinh{2r})=v\exp{-2r}.
    \end{eqnarray}
    \end{widetext}

Thus we arrive at the following:
\begin{itemize}
    \item It is entangled for $e^{-2r}<\frac{1}{v}$.
\end{itemize}

(3) Lastly, let us show that 
\begin{eqnarray}
    \max_{f(\V_{AB})=0} E(\V_{AB}) \approx 0.1882,
\end{eqnarray}
where $E(\cdot)$ represents the logarithmic negativity \cite{vidal02negativity,plenio2005logarithmic} with natural base.

As we have seen in (1) of this proof, $\V_{AB}$ has the zero \gat, $f(\V_{AB})=0$, if and only if $\cosh 2r \le \frac{v+\sqrt{v^2+8}}{4}$, or equivalently, $v\ge \frac{\cosh 4r}{\cosh 2r}$. 

Now suppose that $\V_{AB}$ satisfies $v\ge \frac{\cosh 4r}{\cosh 2r}$, that is, $f(\V_{AB})=0$. Let $E(\V_{AB})$ be the logarithmic negativity of $\V_{AB}$, i.e., $E(\V_{AB}):= -\ln\tilde{v}^{(-)}$, where $\tilde{v}^{(-)}$ represents the smallest value in symplectic spectrum of $\Tilde{\V}_{AB}\equiv(\I\oplus \mathbf{Z})\V_{AB}(\I\oplus \mathbf{Z})$. Then, we can show that 
\begin{eqnarray}
    E(\V_{AB})&\equiv& -\ln(\tilde{v}^{(-)}) = -\ln v + \ln e^{2r}\\
    &\le& -\ln \frac{\cosh 4r}{\cosh 2r} + 2r.
\end{eqnarray}

The terms in the last line take its maximum, denoted by $E_{max}$, around $r\approx0.2203$ and $E_{max}\approx0.1882$. 

\end{proof}

\cref{fig:ex2_vr} summarizes the \gat\ and entanglement of the states with respect to $v,r$.\\
\begin{figure}
    \centering
    \includegraphics[scale=0.5]{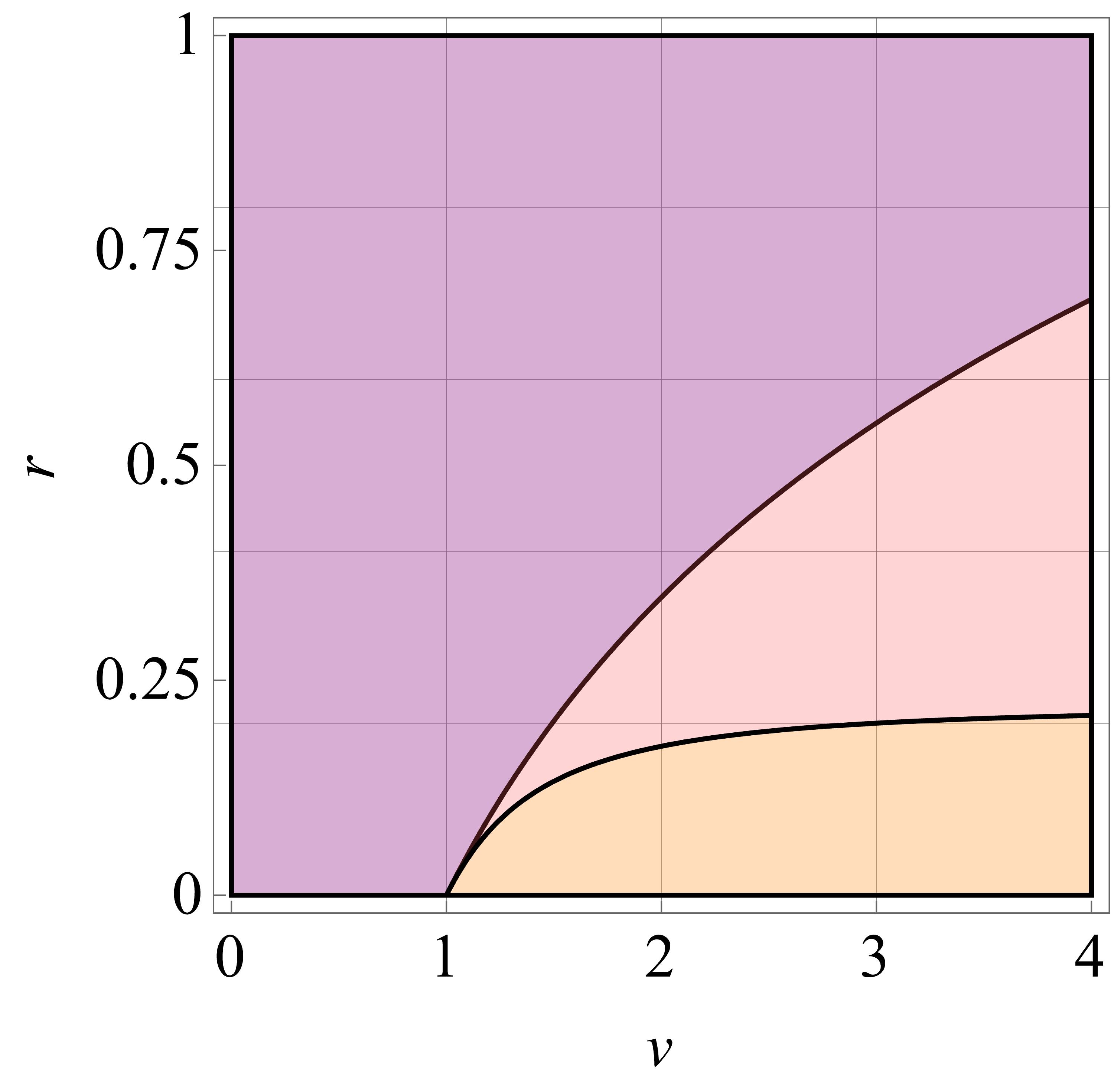}
    \caption{\textbf{A map of states with respect to $v,r$.} States in the yellow region are separable and not Gaussian-atemporal. States in the light pink region are entangled and not Gaussian-atemporal. The states in the rest (purple) are entangled and Gaussian-atemporal. }
    \label{fig:ex2_vr}
\end{figure}

\textit{Example 2.} Consider two Gaussian states $A,B$ whose CMs are $\V_A= (\begin{smallmatrix}
    u & 0\\
    0 & v
\end{smallmatrix})$, $\V_B= (\begin{smallmatrix}
    v & 0\\
    0 & u
\end{smallmatrix})$, respectively, passing through a balanced beam splitter, where $uv\ge1$. Then, the CM $\V_{AB}$ of the form 
\begin{eqnarray}
    \V_{AB} = \frac{1}{2}\begin{pmatrix}
    u+v & 0 & -(u-v) & 0\\
    0 & u+v & 0 & u-v\\
    -(u-v) & 0 & u+v & 0\\
    0 & u-v & 0 & u+v
\end{pmatrix} 
\end{eqnarray} 
represents a general form of the resultant symmetric states.
The positive partial transpose (PPT) criterion \cite{peres96ppt, horodecki96ppt} allows us to find the simple condition for entanglement; it is entangled if and only if $\min(u,v)<1$. By using \gar, we can also find that it is Gaussian-atemporal if and only if 
\begin{eqnarray}
    \frac{1-u}{u^2}>\frac{v-1}{v^2}.
\end{eqnarray}

\begin{proof}
    As $\V_{AB}$ is invariant under the action of switching the role of $A$ and $B$, we have $f(\V_{AB})=\overleftarrow{f}(\V_{AB})=\overrightarrow{f}(\V_{AB})$, thus we will use the forward \gar\ to derive the condition for \gat. The forward pseudo-channel of $\V_{AB}$ is given in a form of its transformation matrix and noise matrix,
    \begin{eqnarray}
        \T = \C\t\V_A^{-1} = -\frac{u-v}{u+v}\mathbf{Z},\\
        \N = \V_{B}-\C\t\V_A^{-1}\C = \frac{2uv}{u+v}\I,
    \end{eqnarray}
    where $\mathbf{Z}\equiv\begin{pmatrix}
        1 & 0\\ 0 & -1
    \end{pmatrix}$ and $\I\equiv\begin{pmatrix}
        1 & 0\\ 0 & 1
    \end{pmatrix}$.
    Then, we can determine the forward atemporality matrix as $\X=\N+i(\bOmega-\T\bOmega\T\t)$. By using \cref{thm:arobust}, a forward \gar\ can be written as  
    \begin{eqnarray}
        \overrightarrow{f}(\V_{AB}) = \max(0, \abs{\omega}-\sqrt{n_1n_2}),
    \end{eqnarray}
    with $\omega = 1-\frac{\det\C}{\det \V_A}$ and eigenvalues $n_1,n_2 \ge 0$ of $\N$. Having $\T,\N$ found, we can readily find that $w=\frac{2(u^2+v^2)}{(u+v)^2}$ and $n_1=n_2=\frac{2uv}{u+v}$. Then it follows that 
    \begin{eqnarray}
        \overrightarrow{f}(\V_{AB}) &=& \max(0, 2\frac{u^2(1-v)+v^2(1-u)}{(u+v)^2}) \\
        &=& \max(0, \frac{2}{(u+v)^2u^2v^2}\left(\frac{(1-v)}{v^2}+\frac{(1-u)}{u^2}\right)).\nonumber
    \end{eqnarray}

    Thus, we can find the condition for $\overrightarrow{f}(\V_{AB})>0$ as follows
    \begin{eqnarray}
        \overrightarrow{f}(\V_{AB})>0 \Leftrightarrow \frac{(1-v)}{v^2}+\frac{(1-u)}{u^2} >0,
    \end{eqnarray}
    since $\frac{2}{(u+v)^2u^2v^2}$ is always positive.
    
    In turn, let us take a look at the conditions for entanglement. In the case of two-mode Gaussian states, the PPT condition becomes the necessary and sufficient condition for entanglement, and thus we have that
    $\tilde{v}^{(-)}<1$ if and only if the state is entangled, where $\tilde{v}^{(-)}$ represents the minimum value in the symplectic spectrum of $\Tilde{\V}_{AB}\equiv (\I\oplus \mathbf{Z})\V_{AB}(\I\oplus\mathbf{Z})$ \cite{peres96ppt, horodecki96ppt,werner01bound,serafini03symplectic}. We will complete the proof by showing that $\tilde{v}^{(-)} = \min(u,v)$.
    
    $\tilde{v}^{(-)}$ is the smallest value in symplectic spectrum of
    \begin{eqnarray}
        \Tilde{\V}_{AB} &=& (\I\oplus \mathbf{Z})\V_{AB}(\I\oplus\mathbf{Z})\\
        &=& \frac{1}{2}\begin{pmatrix}
    u+v & 0 & -(u-v) & 0\\
    0 & u+v & 0 & -(u-v)\\
    -(u-v) & 0 & u+v & 0\\
    0 & -(u-v) & 0 & u+v
\end{pmatrix} \\
&=:& \begin{pmatrix}
    \mathbf{A} & \mathbf{C} \\
    \mathbf{C} & \mathbf{B} 
\end{pmatrix},
    \end{eqnarray} and then it can be found by using the following known formula,
    \begin{eqnarray}
        \tilde{v}^{(-)} = \sqrt{\frac{\Tilde{\Delta}-\sqrt{\Tilde{\Delta}^2-4\det{\Tilde{\V}_{AB}}}}{2}},
    \end{eqnarray}
    where $\Tilde{\Delta} \equiv \mathbf{A} + \mathbf{B} + 2\mathbf{C}$. Here $\Tilde{\Delta}$ yields 
    \begin{eqnarray}
        \Tilde{\Delta} = \frac{(u+v)^2}{4} + \frac{(u+v)^2}{4} + \frac{2(u-v)^2}{4} = u^2+v^2,
    \end{eqnarray}
    and $\det\Tilde{\V}_{AB}$ yields 
    \begin{eqnarray}
        \det\Tilde{\V}_{AB} &=& \frac{(u+v)}{2} \begin{vmatrix}
            \frac{(u+v)}{2} & 0 & \frac{-(u-v)}{2} \\
            0 & \frac{(u+v)}{2} & 0 \\
            \frac{-(u-v)}{2} & 0 & \frac{(u+v)}{2}
        \end{vmatrix}\nonumber\\
        &&-\frac{(u-v)}{2} \begin{vmatrix}
            0 & \frac{(u+v)}{2} & \frac{-(u-v)}{2}\\
            \frac{-(u-v)}{2} & 0 & 0\\
            0 & \frac{-(u-v)}{2} & \frac{(u+v)}{2}
        \end{vmatrix}\\
        &=& \frac{(u+v)^2}{4}\left(\frac{(u+v)^2}{4}-\frac{(u-v)^2}{4}\right) \nonumber\\
        &&- \frac{(u-v)^2}{4} \left(\frac{(u+v)^2}{4} - \frac{(u-v)^2}{4}\right)\\
        &=& \frac{1}{16}\left((u+v)^2-(u-v)^2\right)\left((u+v)^2-(u-v)^2\right) \nonumber\\
        &=& \frac{1}{16}(4uv)(4uv) = u^2v^2.
    \end{eqnarray}
    Having $\Tilde{\Delta}=u^2+v^2$ and $\det\Tilde{\V}_{AB}=u^2v^2$, we have 
    \begin{eqnarray}
        \tilde{v}^{(-)} &=& \sqrt{\frac{\Tilde{\Delta}-\sqrt{\Tilde{\Delta}^2-4\det{\Tilde{\V}_{AB}}}}{2}}\\
        &=& \sqrt{\frac{u^2+v^2-\sqrt{(u^2+v^2)^2-4u^2v^2}}{2}}\\
        &=& \sqrt{\frac{u^2+v^2-\abs{u^2-v^2}}{2}} \\
        &=& \begin{cases}
            \sqrt{\frac{2v^2}{2}} \quad \text{if } u\ge v,\\
            \sqrt{\frac{2u^2}{2}} \quad \text{if } u < v,
        \end{cases} \\
        &=& \begin{cases}
            v \quad \text{if } u\ge v,\\
            u \quad \text{if } u < v,
        \end{cases} \\
        &=& \min(u,v).
    \end{eqnarray}

\end{proof}

\textit{Example 3.} Imagine a Gaussian state whose CM is $\V_A= ( \begin{smallmatrix}
v_1 & 0 \\
0 & v_2
\end{smallmatrix} )$, evolving through a loss channel $\mathcal{L}:A \to B$ such that it is characterized by $\T=\sqrt{\eta}\I, \N = (1-\eta)\I$ with $v_1v_2\ge1$ and a transmission rate $\eta < 1$. By construction, a space-time CM, $\V_{AB}$, is not forward Gaussian-atemporal. However, its reverse Gaussian pseudo-channel is not always a valid channel, which leads to a non-zero reverse Gaussian-atemporality. Indeed, $\V_{AB}$ is reverse Gaussian-atemporal, when $v_1^{(\eta)}v_2^{(\eta)} < 1$, where $v_k^{(\eta)}\equiv\frac{v_k}{1+(v_k-1)\eta}$.

\begin{proof}
    From the scenario stated above, we can construct a space-time CM by using \cref{lem:temporal} as  
    \begin{eqnarray}
        \V_{AB}^\star = \begin{pmatrix}
    v_1 & 0 & v_1\sqrt{\eta} & 0\\
    0 & v_2 & 0 & v_2\sqrt{\eta}\\
    v_1\sqrt{\eta} & 0 & 1+(v_1-1)\eta & 0\\
    0 & v_2\sqrt{\eta} & 0 & 1+(v_2-1)\eta
\end{pmatrix}.    
    \end{eqnarray}
    For notational brevity, we will omit $^\star$ in the rest of the proof. By switching the role of $A$ and $B$ and using \cref{thm:cv_retrieve}, we find the reverse pseudo-channel from it, which is characterized by $\Tilde{\T}, \Tilde{\N}$:
    \begin{eqnarray}
        \Tilde{\T} &=& \C\t\V_B^{-1}\\
            &=& \begin{pmatrix}
    \frac{v\sqrt{\eta}}{1+(v_1-1)\eta} & 0 \\
    0 & \frac{v\sqrt{\eta}}{1+(v_2-1)\eta} 
\end{pmatrix},
    \end{eqnarray}
    and 
    \begin{eqnarray}
        \Tilde{\N} &=& \V_A - \C\t\V_B^{-1}\C,\\
            &=& \begin{pmatrix}
v_1 & 0 \\
0 & v_2
\end{pmatrix}-
\begin{pmatrix}
\frac{v_1^2\eta}{1+(v_1-1)\eta} & 0 \\
0 & \frac{v_2^2\eta}{1+(v_2-1)\eta}
\end{pmatrix}\\
            &=& \begin{pmatrix}
\frac{v_1(1-\eta)}{1+(v_1-1)\eta} & 0 \\
0 & \frac{v_2(1-\eta)}{1+(v_2-1)\eta}
\end{pmatrix}.
    \end{eqnarray}

    Having $\Tilde{\T}, \Tilde{\N}$ found, we can readily have the reverse \gar $\overleftarrow{f}$ as follows:
    \begin{widetext}
        \begin{eqnarray}
        \overleftarrow{f}(\V_{AB}) &=& \max(0, 1-\frac{v_1v_2\eta}{(1+(v_1-1)\eta)(1+(v_2-1)\eta)} - \sqrt{(\frac{v_1(1-\eta)}{1+(v_1-1)\eta})(\frac{v_2(1-\eta)}{1+(v_2-1)\eta}}))\\
        &=& \max(0, 1-v^2\eta-v(1-\eta))\\
        &=& \max(0, (\eta v+1)(1-v)),
    \end{eqnarray}
    \end{widetext}
    where $v=\sqrt{v_1^{(\eta)}v_2^{(\eta)}}$ and $v_k^{(\eta)}\equiv\frac{v_k}{1+(v_k-1)\eta}$. Note that $v>0$.
    
    The reverse \gar, $\overleftarrow{f}(\V_{AB})$, is then non-zero if and only if $1-v >0$, because $\eta>0$ and $v>0$. That is, $\overleftarrow{f}(\V_{AB})>0$ if and only if $v_1^{(\eta)}v_2^{(\eta)} < 1$. Hence, we conclude that $\V_{AB}$ is reverse Gaussian-atemporal, when $v_1^{(\eta)}v_2^{(\eta)} < 1$.
 \end{proof}   

\bibliography{apssamp}

\begin{thebibliography}{52}
\bibitem [{Not()}]{Note_process}%
  \BibitemOpen
  \href@noop {} {}\bibinfo {note} {A general causal distribution has been
  regarded as the one described by a process matrix
  \cite{oreshkov12nocausalorder}, which is capable of connecting systems
  without a definite causal order. To give a clear motivation of the scenario
  that we are concerned in an operational manner, we depict a general causal
  distribution mechanism in Fig.~1 with a process of which the causal order of
  observers' interventions is well-defined. Nonetheless, this is merely for the
  sake of a clear description of experimental setup, and the results presented
  in this paper hold true even for the most general causal mechanism described
  by any process matrix.}\BibitemShut {Stop}%
\end{thebibliography}

\end{document}